\documentclass{article}
\usepackage{arxiv}

\usepackage[utf8]{inputenc} 
\usepackage[T1]{fontenc} 
\usepackage{hyperref}       
\usepackage{url}            
\usepackage{booktabs}       
\usepackage{amsfonts}       
\usepackage{amsmath}        
\usepackage{amsthm}

\usepackage{bbm}       
\usepackage{algorithm}
\usepackage[noend]{algpseudocode}

\usepackage{xfrac}
\usepackage{pgf,tikz,pgfplots}
\pgfplotsset{compat=1.14}
\usepackage{mathrsfs}
\usetikzlibrary{arrows}

\usepackage{pifont} 
\newcommand{\cm}{\ding{51}}
\newcommand{\xm}{\ding{55}}

\DeclareMathOperator*{\E}{\mathbb{E}} 
\DeclareMathOperator{\chain}{\text{chain}} 
\DeclareMathOperator{\UMMB}{\text{U-MMB}} 
\DeclareMathOperator{\FMMB}{\text{F-MMB}} 
\DeclareMathOperator{\MMB}{\text{MMB}} 
\DeclareMathOperator{\UMMR}{\text{U-MMR}}
\DeclareMathOperator{\FMMR}{\text{F-MMR}}
\DeclareMathOperator{\MMR}{\text{MMR}} 

\newtheorem{theorem}{Theorem}
\newtheorem{corollary}[theorem]{Corollary}
\newtheorem{lemma}[theorem]{Lemma}

\title{The Merkle Mountain Belt} 

\author{
  Alfonso Cevallos \\
  Cryp GmbH, Switzerland \\
  \texttt{alfonsoc@gmail.com} \\
   \And 
 Robert Hambrock \\
 Cryp GmbH, Switzerland \\
 \texttt{robert@hambro.ch}
   \And
 Alistair Stewart \\
  Web 3.0 Technologies Foundation, Switzerland \\
  \texttt{alistair@web3.foundation} \\
}

\begin{document}
\maketitle

\begin{abstract}
\emph{Merkle structures} are widely used as commitment schemes: 
they allow a prover to publish a compact commitment to an ordered list $X$ of items, 
and then efficiently prove to a verifier that $x_i\in X$ is the $i$-th item in it. 
We compare different Merkle structures and their corresponding properties as commitment schemes 
in the context of blockchain applications. 
Our primary goal is to speed up \emph{light client protocols} so that, e.g., a user can verify a transaction efficiently from their smartphone. 

For instance, the \emph{Merkle Mountain Range} (MMR) yields a \textbf{succinct} scheme: a light client synchronizing for the first time can do so with a complexity sublinear in $|X|$. 
On the other hand, the \emph{Merkle chain}, traditionally used to commit to block headers, is not succinct, 
but it is \textbf{incremental} --~a light client resynchronizing frequently can do so with constant complexity~-- and \textbf{optimally additive} --~the structure can be updated in constant time when a new item is appended to list $X$. 

We introduce new Merkle structures, most notably the \emph{Merkle Mountain Belt} (MMB),  
the first to be simultaneously succinct, incremental and optimally additive. 
A variant called $\UMMB$ is also \textbf{asynchronous}: a light client may continue to interact with the network even when out of sync with the public commitment. 
Our Merkle structures are slightly unbalanced, so that items recently appended to $X$ receive shorter membership proofs than older items. 
This feature reduces a light client's expected costs, in  applications where queries are biased towards recently generated data. 


\end{abstract}

\section{Introduction}

Blockchain technology has gained considerable interest in both academia and industry. 
It allows for the first time for end-to-end transfers of value on a global scale without the need for trusted intermediaries. 
This technology has also sparked renewed interest in cryptographic primitives such as zero-knowledge proofs, signature schemes, and commitment schemes. 

However, in its basic architecture, a public blockchain is trustless only for \emph{full nodes}, i.e., participants who are permanently online, store the blockchain state in full and follow the consensus protocol. 
The trustless property is typically lost for \emph{light clients}, i.e., participants who are seldom online
and have limited capabilities in terms of bandwidth, computation and memory. 
A light client could represent a user's personal device, an Internet of Things (IoT) device, or a smart contract connected via a cross-chain bridge to a different network. 

For users, this presents a dilemma between safety and convenience. 
For example, remote procedure call (RPC) services, which are centralized intermediaries that relay specific data from a blockchain network under request, often do not provide their users with any means to authenticate the data delivered. 
Consequently, if this service is compromised, the user of a self-custody wallet might be shown inaccurate balances and other deceptive details, potentially leading the user to make transactions they would otherwise avoid, or avoid transactions they would otherwise make~\cite{Sorgente}. 

Similarly, many cross-chain communication bridges are maintained by centralized entities that, again, need to be trusted by their users. 
Unsurprisingly, in recent years cross-chain bridges have been targets of numerous attacks; see surveys \cite{lee2023sok, belenkov2025sok}. 
As per a 2022 report by Chainalysis~\cite{Chainalysis}, these bridge attacks account for over 60\% of all funds stolen in cryptocurrency hacks, leading to losses that surpass 2 billion USD so far.

To address these issues, several light-client protocols (LCPs) have recently been proposed \cite{kiayias2020non, bunz2020flyclient, agrawal2023proofs, bhatt2025trustless}, both for proof-of-work (PoW) and proof-of-stake (PoS) blockchains. 
An LCP should provide light client Alice with efficient ways to 
a)~synchronize with the network and, in particular, identify a public commitment (or digest) of its current state, and then 
b)~retrieve some specific data of interest from an untrusted full node Bob, and verify its authenticity against said commitment. 
These protocols require Bob to maintain \emph{commitment schemes}~\cite{benaloh1993one}, such as \emph{Merkle structures} (defined below) over different data sets of the network. 

In this paper we explore the design space of Merkle structures as commitment schemes for blockchain applications, and identify several desirable properties that help speed up LCPs. 
We then propose new structures that achieve, for the first time, all the identified properties simultaneously. 
In particular, our $\MMB$ structure combines the best features of an $\MMR$ and a Merkle chain, which are commonly used in blockchain (and defined below); see Table~\ref{tab:abstract}. 

\begin{table}[htb]
    \centering
    \begin{tabular}{|l|c|c|c|}
    \hline
    \textbf{} & \textbf{chain} & \textbf{MMR} & \textbf{MMB} \\
    \hline 
    Update time per item append & $O(1)$ & $O(\log n)$ & $O(1)$ \\
    Proof size of $k$-th most recent item & $O(k)$ & $O(\log n)$ & $O(\log k)$ \\
    \hline
    \end{tabular}
    \caption{Comparison of MMB to the Merkle chain and to MMR~\cite{MMR}, assuming $1\leq k\leq n=|X|$.}
    \label{tab:abstract}
\end{table}

\vspace{5mm}

While a \emph{Merkle tree} \cite{merkle1989certified, camacho2012strong} is assumed to be balanced, a \emph {Merkle structure} is a generalization composed of one or more possibly unbalanced trees. 
Formally, given a list $X=(x_1, \cdots, x_n)$ and public hash function $H(\cdot)$, a Merkle structure is a collection of one or more rooted binary trees where each leaf $h_i=H(x_i)$ is identified with the hash (evaluation) of an item $x_i$ in~$X$, each non-leaf node is identified with the hash of the concatenation of its two children, and finally the list of root hashes acts as a public commitment $\langle X\rangle$ to list $X$. 
As a commitment scheme, it allows a prover to build a \emph{membership proof} $\pi_{x_i\in X}$ for any item $x_i$, which a verifier can compare against commitment $\langle X\rangle$ to attest that $x_i$ is the $i$-th item in $X$; see Figure~\ref{fig:membership}. 

\begin{figure}[htp]
\centering
\includegraphics[width=0.85\textwidth]{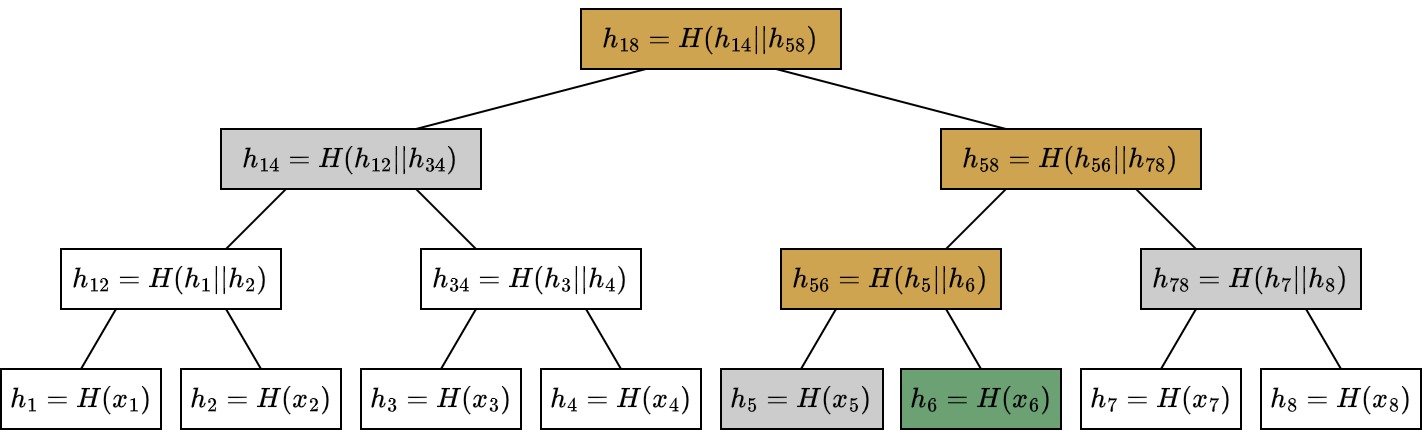}
\caption{A Merkle tree for list $X=(x_1, \cdots, x_8)$ and hash function $H(\cdot)$. 
For item $x_6$, its leaf $h_6$ is in green, and its membership proof is 
$\pi_{x_6 \in X}=((h_5, \text{left}), (h_{78}, \text{right}), (h_{14}, \text{left}))$, 
formed by the leaf's ancestors' siblings, in gray. 
From such an item and list, a verifier can reconstruct the actual ancestors, in brown, and ultimately the root, which should match public commitment $\langle X\rangle$.}
\label{fig:membership}
\end{figure}

The size of proof $\pi_{x_i\in X}$ is proportional to the distance from leaf $h_i$ to its root, so it is $O(\log n)$ in the case of a balanced Merkle tree, where $n=|X|$ is the list size. 
For instance, such a tree is used to commit to the \emph{static} set of transactions contained in a block, and its root hash is included in the block header. 
This way, Alice can retrieve from Bob a single transaction, along with a logarithmic-sized proof, and authenticate it against the header. 

However, in this work we focus on \emph{dynamic} applications, with an append-only list $X$ 
that is continually growing. 
A prime example is committing to the list of all block headers, from genesis to the latest authored block,  
while other blockchain applications include committing to the list of transaction outputs (TXOs) in a UTXO-based blockchain~\cite{chen2020minichain}, 
and the list of previous validator committees in a PoS blockchain~\cite{agrawal2023proofs}.
For such applications, we study \emph{additive} commitment schemes \cite{baldimtsi2017accumulators, icissp21}: 
those that allow for efficient item appends.

The \emph{Merkle Mountain Range} (MMR)~\cite{MMR} has gained popularity as an additive scheme with $O(\log n)$ time appends and $O(\log n)$ sized proofs. 
It is at the core of many recent \emph{succinct} (or superlight) LCPs \cite{bunz2020flyclient, moshrefi2021lightsync, lan2021horizon, agrawal2023proofs, tas2024light, bhatt2025trustless}: those where light clients' operations have a sublinear complexity relative to the size of the queried data set. 
In turn, we say that a commitment scheme is \textbf{succinct} if any verifier's query has a complexity sublinear in $n$. 
We formalize this definition later in this section, but note here that all LCPs mentioned above continue to be succinct if they switch from MMR to any other succinct Merkle structure.

However, the \emph{Merkle chain} remains the most widely used scheme for committing to the list of block headers in a blockchain. 
\emph{Indeed, it is at the origin of the term ``blockchain''.} 
The Merkle chain corresponds to a maximally unbalanced tree, and is therefore far from succinct; see Figure~\ref{fig:chain}. 
It is used by most networks, including Bitcoin~\cite{nakamoto2019bitcoin} and Ethereum~\cite{wood2014ethereum}, and is at the core of Simple Payment Verification (SPV)~\cite{nakamoto2019bitcoin}, an LCP proposed in the original Bitcoin paper that requires a light client to download the headers of all existing blocks. 
As of 2025, an Ethereum SPV client would need to download more than 10~GB of data in headers, which makes it unfit for a mobile application.%
\footnote{As of late 2025, there are over 23 million blocks produced on Ethereum. 
Their headers are of variable size, but we estimate their average size to be at least $0.5$ KB.} 

\begin{figure}[htp]
\centering
\includegraphics[width=0.65\textwidth]{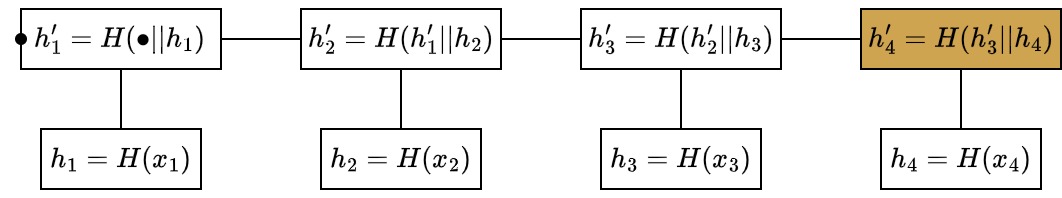}
\caption{In this Merkle chain, the root is in brown, and symbol $\bullet$ visually represents a non-existent child, and a default hash value for it. 
A new item $x_5$ can be appended to $X$ in constant time: add a new leaf $h_5=H(x_5)$, and create a new root $h'_{5}=H(h'_{4} || h_5)$ as the parent of $h'_{4}$ and $h_5$.}
\label{fig:chain}
\end{figure}

Why does the Merkle chain continue to be popular, despite its shortcomings? 
A key advantage is that it yields an \textbf{optimally additive} scheme: 
it takes constant time to update after an item append.  
This is important for full nodes, who update the scheme continually. 
In particular, some LCPs~\cite{bunz2020flyclient, moshrefi2021lightsync} require each new block header to contain the latest commitment to all previous headers, so this update must be executed during block production. 

Let $X_n$ be the state of list $X$ when it contains $n$ items, 
and let a \emph{$k$-increment} 
be the event that $k$ items are appended to $X$, so its state changes from $X_m$ to $X_n$ with $k=n-m$. 
The~Merkle chain is also \textbf{incremental}: if light client Alice goes offline during a $k$-increment, the cost she incurs when synchronizing again depends \emph{only on~$k$ and not on~$n$}. 
This property ensures that if Alice resyncs often enough, her cost per sync remains bounded by a constant. 

An incremental scheme is also \emph{checkpoint friendly}: 
a checkpoint is a commitment $\langle X_m \rangle$ to a recent state $X_m$ that Alice can retrieve from a trusted source when she first joins the network. 
This property allows her to synchronize from this checkpoint to the latest state $X_n$ much faster than if she did it from genesis. 
For example, it has been proposed~\cite{tas2023bitcoin} that PoS networks may speed up LCPs by periodically time-stamping such checkpoints onto Bitcoin.


Finally, we remark that some LCPs~\cite{bunz2020flyclient, bhatt2025trustless} require schemes with \textbf{increment proofs}:%
\footnote{Also called incremental proofs~\cite{crosby2009efficient}, consistency proofs~\cite{meiklejohn2020think}, subtree proofs~\cite{bunz2020flyclient}, prefix proofs~\cite{tyagi2022versa} and ancestry proofs~\cite{bhatt2025trustless}. 
This notion seems to have a universal appeal within additive commitment schemes.} 
if Alice goes offline during a $k$-increment, so that the list state has changed from $X_m$ to $X_n$, she can request an increment proof $\pi_{X_m\prec X_n}$ from full node Bob, and verify it against both commitments $\langle X_m \rangle$ and $\langle X_n \rangle$ to attest that $X_n$ indeed resulted from a $k$-increment over~$X_m$, or in other words, that $X_m$ is a \emph{list prefix} of~$X_n$ ($X_m\prec X_n$). 
Such a proof allows her to audit Bob, and make sure he has not made changes to old items while she was offline. 

\subsection*{Our contributions}

We consider two key roles related to any commitment scheme: \emph{participants} who resynchronize periodically after offline periods, request and verify proofs, and possibly store and update proofs as well, and a \emph{manager} who performs appends, keeps the scheme updated, and provides proofs under request. 
We have argued that, for the purpose of committing to an append-only list $X$ in a blockchain, and enabling efficient LCPs, a commitment scheme should ideally be able to issue membership and increment proofs, and observe the following properties:

\begin{itemize} 
    \item \textbf{Succinctness:} For any participant, each of the following operations has a communication and computational complexity (CCC) polylogarithmic in $n=|X|$: 
    \begin{itemize}
        \item retrieving the current commitment $\langle X_n \rangle$,
        \item retrieving and verifying any membership proof $\pi_{x_i\in X_n}$, $1\leq i\leq n$ and
        \item retrieving and verifying any increment proof $\pi_{X_m\prec X_n}$, $1\leq m< n$.
    \end{itemize}
    \item \textbf{Incrementality:} For a participant resynchronizing from $X_m$ to $X_n$ with $k=n-m$, each of the following operations has a CCC that depends \emph{only} on $k$ and not on $n$: 
    \begin{itemize} 
        \item updating their copy of the commitment from $\langle X_m \rangle$ to $\langle X_n \rangle$, 
        \item retrieving and verifying the increment proof $\pi_{X_m\prec X_n}$,
        \item retrieving and verifying a new membership proof $\pi_{x_i\in X_n}$, $m<i\leq n$, and
        \item updating and verifying an old membership proof from $\pi_{x_i\in X_m}$ to $\pi_{x_i\in X_n}$.
    \end{itemize}
    \item \textbf{Optimal additivity:} Appending an item, and performing the corresponding update to the scheme and its commitment $\langle X \rangle$, takes the scheme manager $O(1)$ time.
\end{itemize}

Our contributions include formalizing increment proofs for any append-only Merkle structure --~which we do in Section~\ref{s:increment}~-- as well as introducing the incrementality property above where, in particular, we establish four basic \emph{resync operations} widely used in LCPs. 
While the Merkle chain observes the last two properties above, it fails the first one, and vice-versa for MMR. 
This poses a dilemma for protocol designers in terms of choice of Merkle structure. 
We solve this dilemma with the Merkle Mountain Belt (MMB): 

\begin{theorem}\label{thm:MMB}
MMB is an incremental, succinct and optimally additive commitment scheme. 
For $n=|X|$, it produces a constant-sized commitment, the manager needs $O(\log n)$ memory and $O(1)$ time per append, 
and a participant needs $O(\log n)$ memory and $O(\log k)$ CCC to perform any of the four basic resync operations above after a $k$-increment. 
\end{theorem}

Chatzigiannis et al.~remarked in a recent survey~\cite[Gap~3]{chatzigiannis2022sok} that so far no LCP has specifically considered how a light client can resync, after an offline phase, faster than synchronizing for the first time. 
MMB addresses this gap with its incrementality property. 
In fact, we contend that $\MMB$ is \textit{the most fitting Merkle structure yet}, for a variety of dynamic applications that currently use either a Merkle chain or $\MMR$ to commit to an append-only list; see Table~\ref{tab:apps}.

\begin{table}[htp]
\centering
\begin{tabular}{|l|l|l|}
\hline
\textbf{Protocol} & \textbf{List $X$} & \textbf{Main improvement} \\ 
\hline 
\begin{tabular}{@{}l@{}} FlyClient~\cite{bunz2020flyclient}, \\ e.g., in Zcash~\cite{hopwoodzcash} \end{tabular}
& Block headers & 
\begin{tabular}{@{}l@{}} CCC to resync after a $k$-increment \\ reduces from $O(\log n \cdot \log k)$ to $O(\log^2 k)$ \end{tabular}  \\
\hline 
\begin{tabular}{@{}l@{}} Proofs of  \\ Proof-of-Stake~\cite{agrawal2023proofs} \end{tabular}
& Validator committees & 
\begin{tabular}{@{}l@{}} CCC to resync after a $k$-increment \\ reduces from $O(\log n)$ to $O(\log k)$ \end{tabular}  \\
\hline 
\begin{tabular}{@{}l@{}} Cross-chain bridge,  \\ 
e.g., in Polkadot~\cite{bhatt2025trustless} \end{tabular} &
Outgoing messages  & 
\begin{tabular}{@{}l@{}} Expected MP size (\& target-network fees) \\ 
probably reduce from $O(\log n)$ to $O(1)$ \end{tabular} \\  
\hline 
MiniChain~\cite{chen2020minichain} & 
\begin{tabular}{@{}l@{}} Transaction outputs \\ in UTXO blockchain \end{tabular} & 
\begin{tabular}{@{}l@{}} Expected MP size (\& transaction fees) \\ probably reduce from $O(\log n)$ to $O(1)$ \end{tabular}  \\
\hline
\begin{tabular}{@{}l@{}} Registration-based  \\ 
encryption~\cite{garg2018registration} \end{tabular} &
\begin{tabular}{@{}l@{}} (Identity, public key) \\ 
pairs, one per user  \end{tabular}  & 
\begin{tabular}{@{}l@{}} CCC of registration \& decryption key \\ update reduces from $O(\log n)$ to $O(1)$ \end{tabular} \\ \hline 
\end{tabular}
\caption{
Improvements on protocols that use $\MMR$~\cite{MMR} or $\UMMR$~\cite{reyzin2016efficient}, if they switch to $\MMB$ (resp.~$\UMMB$), assuming $1\leq k \leq n=|X|$. 
In all of them, the update time per append is also reduced from $O(\log n)$ to $O(1)$. 
MP stands for membership proof. 
See Appendix~\ref{s:apps} for details. 
}
\label{tab:apps}
\end{table}

Next, we consider applications where light clients may need to interact with the network while holding outdated versions of the commitment and/or membership proofs. 
Examples are applications with a \emph{stateless architecture}, 
as well as \emph{registration-based encryption}; 
see Appendix~\ref{s:apps} for details.  
In this context, Reyzin and Yakoubov~\cite{reyzin2016efficient} introduced a commitment scheme with the following \textbf{asynchrony} properties: 

\begin{itemize}
    \item \textbf{Old-commitment compatibility:} The membership of an item $x_i\in X$ can be verified with proof $\pi_{x_i\in X_n}$ against an older commitment $\langle X_m \rangle$, as long as $i\leq m\leq n$, and
        \item \textbf{Low update frequency:} During a $k$-increment, any membership proof needs to be updated a number of times sublinear in $k$.
\end{itemize}

These properties enable Alice to participate even while out of sync: the former property ensures she can use her outdated copy of the commitment to verify the membership of any items, 
while the latter means she might be able to prove the membership of an item, as her proof might still be valid. We introduce a strengthening of the latter property:

\begin{itemize}
    \item \textbf{Recent-proof compatibility:} For an absolute constant $c>0$, the membership of an item $x_i\in X$ can be verified with proof $\pi_{x_i\in X_n}$ against a newer commitment $\langle X_m \rangle$, as long as $i\leq n\leq m\leq n+c(n-i+1)$.
\end{itemize}

In other words: if Alice updates her proof $\pi_{x_i\in X_n}$ at state $X_n$, it is guaranteed to remain valid throughout a future $k$-increment, where $k= \lfloor c\cdot r \rfloor$ is proportional to the recency $r=(n-i+1)$ of item $x_i$. 
For instance, a parameter $c=1/5$ means that if Alice updates the proof of her item five hours after its creation, this proof should remain valid for at least one hour, assuming a constant append rate; see Figure~\ref{fig:async}.   
While this property implies low update frequency (see Lemma~\ref{lem:recent}), the converse is false, as a low update frequency might still make Alice's proof invalid after just one append, i.e., against $\langle X_{n+1} \rangle$. 

The asynchronous scheme presented in~\cite{reyzin2016efficient} is a variant of MMR that we shall refer to as U-MMR, following our naming convention described in Section~\ref{s:bagging}. 
It produces $O(\log n)$ separate trees and therefore has a logarithmic-sized commitment $\langle X \rangle$ composed of all tree roots.
In turn, we present U-MMB, an asynchronous variant of MMB with a logarithmic-sized commitment, that achieves \emph{all properties discussed so far}: 

\begin{theorem}\label{thm:UMMB}
U-MMB is an incremental, succinct and optimally additive commitment scheme that also achieves old-commitment and recent-proof compatibility with $c=1/5$. 
It has the specifications stated in Theorem~\ref{thm:MMB}, except that it produces a commitment of size $O(\log n)$. 
\end{theorem}

\begin{figure}[htp]
\centering
\includegraphics[width=0.6\textwidth]{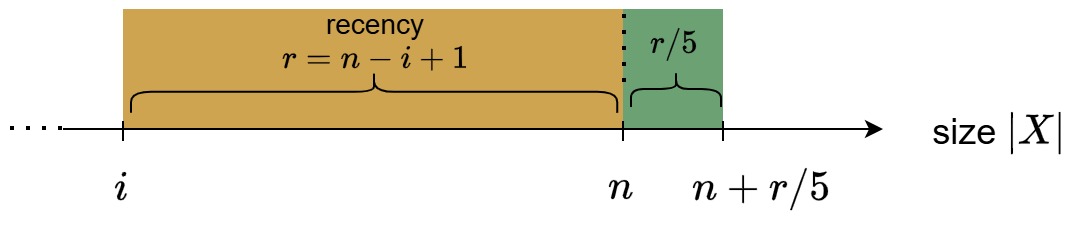}
\caption{In $\UMMB$, membership proof $\pi_{x_i \in X_n}$ is valid against a commitment $\langle X_m \rangle$ for any value of $m$ in the shaded areas: the brown and green areas are guaranteed by old-commitment and recent-proof compatibilities, respectively.}
\label{fig:async}
\end{figure}

Our results are summarized in Tables \ref{tab:abstract} to \ref{tab:results}. 
We note that $\UMMB$ is obtained simply by removing some top-level nodes from $\MMB$. 
Hence, one can easily offer the asynchrony properties of $\UMMB$ to light clients, while full nodes run $\MMB$ and thus store $O(1)$ sized commitments on-chain; 
we discuss this hybrid solution in Appendix~\ref{s:apps}.

\begin{table}[htp]
\centering
\begin{tabular}{|l|c|c c|c c|}
\hline
\textbf{} & \textbf{chain} & \textbf{U-MMR} & \textbf{U-MMB} & \textbf{MMR} & \textbf{MMB} \\ 
\hline
Old-commitment compatibility           & \cm   & \cm       & \cm        & \xm       & \xm \\
Recent-proof compatibility           & \xm   & \xm       & \cm        & \xm       & \xm \\
Low update frequency                   & \xm    & \cm       & \cm        & \xm       & \xm \\  
MP updates during $k$-increment    & $O(k)$   & $O(\log k)$  & $O(\log k)$       & $O(k)$  & $O(k)$ \\ 
\hline
Succinctness                           & \xm    & \cm       & \cm       & \cm       & \cm \\ 
Incrementality                         & \cm   & \xm        & \cm        & \xm       & \cm \\ 
CCC update $\langle X \rangle$ after $k$-increment  & $O(1)$   & $O(\log k)$  & $O(\log k)$       & $O(1)$  & $O(1)$ \\ 
CCC retrieve IP for $k$-increment      & $O(k)$   & $O(\log n)$  & $O(\log k)$  & $O(\log n)$  & $O(\log k)$ \\
CCC retrieve MP $k$-th newest item       & $O(k)$   & $O(\log n)$  & $O(\log k)$  & $O(\log n)$  & $O(\log k)$ \\ 
CCC update MP after $k$-increment      & $O(k)$   & $O(\log n)$  & $O(\log k)$  & $O(\log n)$  & $O(\log k)$ \\
\hline
Optimal additivity                     & \cm   & \xm        & \cm        & \xm       & \cm \\
Manager's append time                   & $O(1)$   & $O(\log n)$  & $O(1)$  & $O(\log n)$       & $O(1)$ \\
Commitment size                         & $O(1)$   & $O(\log n)$  & $O(\log n)$       & $O(1)$  & $O(1)$ \\ 
\hline
\end{tabular}
\caption{Comparison of our schemes to the Merkle chain, U-MMR~\cite{reyzin2016efficient}, and MMR~\cite{MMR}, assuming $1\leq k\leq n=|X|$. 
MP and IP stand for membership proof and increment proof, respectively. 
}
\label{tab:results}
\end{table}

A prominent characteristic of our proposed structures is their imbalance: the leaf of the $k$-th most recently appended item is at distance $\Theta(\log k)$ from its root. 
While a balanced tree may suit static data sets, where all items are equally likely to be queried, our structures are ideal for dynamic applications that continually generate new data, as users' queries will be naturally biased towards recently generated items. 
Hence, providing these items with shorter proofs may lead to savings in expectation.
We turn this observation into a concrete statement:

\begin{theorem}\label{thm:Zipf} 
Consider an application with an append-only list $X$ committed with $\MMB$. 
For a fixed value of $k$ and a variable $n=|X|\geq k$, the expected membership proof size of the $k$-th newest item, measured in hashes, is at most 
$$\frac{11}{8}\log_2\left(\frac{k+1}{3}\right)+\frac{9}{2}-\frac{9}{4(k+1)}.$$ 
\end{theorem}



\vspace{5mm}

The paper is organized as follows. 
We present necessary definitions in Section~\ref{s:prel}. 
Section~\ref{s:unbagged} describes the $\UMMB$ scheme, which readily achieves all our listed properties, but builds $O(\log n)$ separate trees. 
In Section~\ref{s:forward} we discuss common ways to merge or ``bag'' these trees into one, to reach a constant sized commitment, leading to trade-offs in terms of the properties achieved.  
$\MMB$ is presented in Section~\ref{s:double}, with a novel ``bagging'' procedure that retains all listed properties, except for asynchrony. 
We formalize increment proofs in Section~\ref{s:increment}, and we conclude in Section~\ref{s:conc}.  

In Appendix~\ref{s:apps} we present related work, as well as highlight potential applications of our schemes, and we provide algorithmic details in Appendix \ref{s:alg-UMMB}. 
Then, in Appendix \ref{s:amortized} we prove Theorem~\ref{thm:Zipf}, and find exact formulas for the amortized membership proof size of several variants of $\MMB$ and $\MMR$. 
We finish with delayed proofs in Appendix~\ref{s:proofs}.


\section{Preliminaries}\label{s:prel}

Throughout the paper we consider a finite ordered list $X=(x_1, \cdots, x_n)$ onto which new items are continually being appended, but not deleted nor modified. 
All items are assumed to be distinct, appended one at a time, and sorted by increasing append time. 
We denote by $n:=|X|$ the current list size, and by $X_m=(x_1, \cdots, x_m)$ the state of $X$ when it contains $m$ items. 
The relation ``list $Y$ is a prefix of list $Z$'' is denoted as $Y\prec Z$; 
hence, $X_m\prec X_n$ whenever $0\leq m\leq n$. 
We assume items in $X$ come from $\{0,1\}^*$, i.e., they are bit strings of arbitrary finite length, and we represent the concatenation of strings $x$ and $y$ by $x||y$.

We fix a collision-resistant, publicly known hash function $H:\{0,1\}^*\rightarrow \{0,1\}^*$ with images of bounded size. 
A binary Merkle structure $\mathcal{M}$ for $X$ is a collection of one or more rooted binary trees, with a clear bijection between $X$ and the set of all leaves. 
All Merkle structures in this paper can be assumed to be binary, unless specifically noted otherwise. 
We assign to each leaf $h_i=H(x_i)$ the hash (evaluation) of an item $x_i\in X$, and to each non-leaf node the hash of the concatenation of its two children; see Figure~\ref{fig:membership}. 
We also allow for a non-leaf node to have a single child, in which case we represent with symbol $\bullet$ the missing child node (visually) and a default hash value for it; see Figure~\ref{fig:chain}.

This Merkle structure can be used as a commitment scheme~\cite{benaloh1993one} for list $X$, whose public commitment $\langle X \rangle$ corresponds to the list of hashes of the tree roots. 
The membership proof $\pi_{x_i\in X}$ of item $x_i$ is the list of siblings of all ancestors of its leaf $h_i$, along with their \emph{handedness}: whether they are left or right siblings. 
From this information and hash function $H$, a verifier can sequentially compute the hash of each ancestor of $h_i$, and ultimately the hash of its root, and verify whether the latter appears in the public commitment $\langle X \rangle$; see Figure~\ref{fig:membership}. 

In fact, a Merkle structure is a \emph{vector commitment scheme}~\cite{catalano2013vector}, because from the revealed handedness of the ancestors of leaf $h_i$, and the position of its root within $\langle X \rangle$, the verifier can also check the leaf's overall position within the structure, and thus verify the index $i$ of item $x_i$. 
It is computationally infeasible to find a valid membership proof with index $i$ for any item other than $x_i$, as it would imply a hash collision in one of the ancestors of leaf $h_i$. 


\section{U-MMB: an asynchronous scheme}\label{s:unbagged}


Commitment schemes based on Merkle structures are \emph{strong}, meaning that their construction does not require trusted parties or backdoors~\cite{camacho2012strong}. 
As is the case with all known strong schemes, a drawback of Merkle structures is that, generally, an item's membership proof needs to be updated whenever list $X$ changes, as its verification may fail if not in sync with the commitment $\langle X\rangle$. 
In this section, we consider applications where participants cannot easily retrieve the current commitment and/or membership proofs from a central authority. 

We consider a distributed application where several participants collaborate in maintaining the commitment scheme of a common list~$X$, and each locally stores and updates any proofs of interest. 
For simplicity, we assume the presence of a \emph{scheme manager} who updates the scheme after each append and registers all changes in a public \emph{bulletin board}, including the new commitment and the hashes of any new nodes.  
However, this manager is \emph{transparent}: they do not need to be trusted, as any participant can audit their operations with the help of the bulletin board and increment proofs. 
More generally, the manager's duties can be fully replaced with a consensus protocol run by the participants.

As mentioned in the introduction, the $\UMMR$ scheme~\cite{reyzin2016efficient} achieves the \emph{asynchrony} properties of old-commitment compatibility and low update frequency. 
In turn, we strengthen the latter property by introducing recent-proof compatibility.
These properties facilitate the interaction of participants who might be out-of-sync with one another and with the bulletin board. 
$\UMMR$ is succinct but neither incremental nor optimally additive.%
\footnote{Authors in~\cite{reyzin2016efficient} do not describe increment proofs, yet a close inspection reveals that their scheme admits increment proofs of size $O(\log n)$. 
Thus, their scheme is succinct.}
Hence, if $n$ is large and Alice resyncs frequently, she might prefer the use of the Merkle chain over $\UMMR$.  

In this section, we solve this dilemma with \emph{Unbagged MMB} ($\UMMB$), the first succinct, incremental, and optimally additive scheme, that also achieves old-commitment and recent-proof compatibilities. 
Notice in Table~\ref{tab:results} that it matches or outperforms previous Merkle structures in every category, except for the commitment size. 
In later sections, we present new schemes derived from $\UMMB$ that achieve commitments of constant size.

\subsection{Append operation with lazy merges}\label{s:unbagged-description}

Like $\UMMR$~\cite{reyzin2016efficient}, $\UMMB$ is composed of an ordered list of $O(\log n)$ trees. 
Following notation from~\cite{MMR}, we call each of these trees a \emph{mountain}: it has an associated height $s\geq 0$ and is a perfect binary hash tree with $2^s$ leaves, each at distance~$s$ from the mountain's root, which we call its \emph{peak}. 
For instance, Figure~\ref{fig:membership} showcases a height-3 mountain. 
Commitment $\langle X \rangle$ is simply the list of $O(\log n)$ peak hashes. 
Leaves within each mountain correspond to hashes of consecutive items in $X$ sorted by increasing append time, and mountains are likewise sorted so that the rightmost mountain contains the newest leaves; see Figure~\ref{fig:s9}. 
The structure is hence characterized by its sequence $S_n$ of mountain heights, which is determined by~$n$ as we explain below. 
E.g., for $n=9$ the height sequence is $S_9=(2,2,0)$.

\begin{figure}[htp]
\centering
\includegraphics[width=0.85\textwidth]{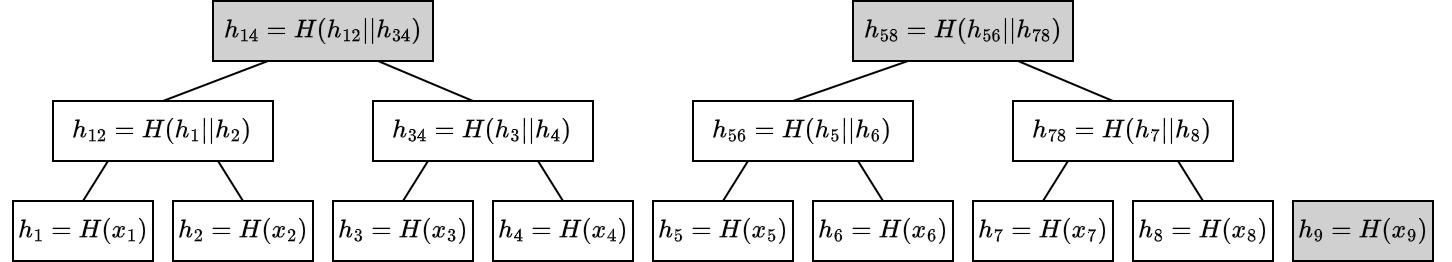}
\caption{$\UMMB$ structure for list $X=(x_1, \cdots, x_9)$, with height sequence $S_9=(2,2,0)$. }
\label{fig:s9}
\end{figure}

Two mountains are \emph{mergeable} if they are consecutive and have the same height~$s$: 
they can be merged into a mountain of height $s+1$ with a single hash computation by creating a \emph{merge peak} 
with the two previous peaks as children. 
In $\UMMR$~\cite{reyzin2016efficient}, mergeable mountains are merged immediately; as a consequence, appending a new item --~by adding its hash as a height-0 mountain~-- can lead to a ``domino effect'' of up to $O(\log n)$ consecutive merge steps. 
Yet, it can be proved that the \emph{average} number of merge steps per append in that scheme is constant: exactly one. 
This inspires us to take a \emph{lazy} approach in $\UMMB$ and execute \emph{at most one} merge step per append, possibly leaving mergeable mountains unmerged. 
That way, we match the amortized and worst-case complexities of the append operation. 

Hence, in $\UMMB$ the \textbf{append operation}, performed by the scheme manager, is:
\begin{enumerate}
    \item \emph{Add step:} add the new leaf as a height-0 mountain to the right end of the list; and
    \item \emph{Merge step:} if there are mergeable mountains, merge the rightmost pair, so that the new mountain is one unit higher and replaces the pair in-place in the mountain list.
\end{enumerate}

\begin{figure}[htp]
\centering
\includegraphics[width=0.85\textwidth]{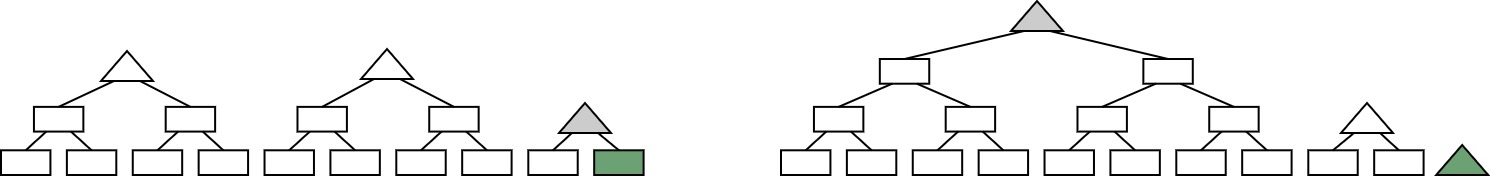}
\caption{Evolution of the U-MMB structure from Figure~\ref{fig:s9} through two item appends, 
where height sequence $S_9=(2,2,0)$ evolves to $S_{10}=(2,2,1)$ and to $S_{11}=(3,1,0)$. 
Peaks are represented with triangles. 
In each append, the new leaf is in green and the merge peak in gray.}
\label{fig:s10s11}
\end{figure}

\begin{lemma}\label{lem:hash-U}
    The $\UMMB$ append operation requires at most one hash computation.
\end{lemma}

See Figure~\ref{fig:s10s11} for an example. 
The proof of the lemma above is skipped, but we highlight that, throughout the paper, we ignore the cost of computing any leaf hash, as we assume that these are given as input to the scheme, and consider only the hash computation of non-leaf nodes. 
Notice that $\UMMB$ matches the Merkle chain in terms of hash computations per append; see Figure~\ref{fig:chain}. 
Finally, we remark that the pair being merged can be found in $O(1)$ time with the help of a \emph{stack} of mergeable mountains.
The manager can then update the commitment $\langle X \rangle$ in the bulletin board, and append to it the hashes of the (one or two) new nodes, all in constant time, and keeping only the $O(\log n)$ peaks in memory. 
In Appendix~\ref{s:alg-UMMB} we present the $\UMMB$ append operation in detail, and prove the following result.

\begin{lemma}\label{lem:U-append}
The $\UMMB$ manager needs $O(\log n)$ memory and $O(1)$ time per append. 
\end{lemma}

\begin{table}[htp]
\centering
\begin{tabular}{|c|c|c|c|c|c|c|c|c|c|c|}
\hline
\textbf{$n$} & $1$ & $2$  & $3$ & $4$ & $5$ & $6$ & $7$ & $8$ & $9$ & $10$ \\ \hline 
$n+1$ in binary & $10$         & $11$ & $100$        & $101$  & $110$  & $111$  & $1000$       & $1001$    & $1010$    & $1011$  \\ 
\hline
\textbf{$S_n$} & $0$ & $1$  & $10$ & $11$ & $20$ & $21$ & $210$ & $211$ & $220$ & $221$ \\ 
\hline 
\end{tabular}
\caption{Height sequence $S_n$ for the first few values of $n$, with commas removed for brevity.}
\label{tab:Sn}
\end{table}

While the add step increases the mountain count by one, the merge step, if executed, decreases the count by one. 
Thus, the net mountain count never decreases. 
In Table~\ref{tab:Sn} we show the sequence $S_n$ of mountain heights for the first few values of~$n$, and compare it against the binary representation of $n+1$. 
By staring long enough at this table, the astute reader may convince themselves of the following properties; we include proofs in Appendix~\ref{s:proofs}.

\begin{lemma}\label{lem:Sn}
$\UMMB$ and its mountain height sequence $S_n$ have the following properties:
\begin{enumerate}
    \item If there are $n$ leaves, the number of mountains is $t:=\lfloor \log_2 (n+1) \rfloor$. 
    \item If $n+1$ in binary is $(b_t\cdots b_1 b_0)$, and we enumerate the mountain heights from right to left starting from zero, i.e., $S_n=(s_{t-1}, \cdots, s_1, s_0)$, then $s_i=b_i+i$ for each $0\leq i<t$.  
    \item During the $n$-th item append, the merge step is skipped if and only if $n+1$ is a power of two. 
    Else, if $j$ is the index of the lowest ``one'' in $(b_t\cdots b_1 b_0)$, then we merge two mountains of height $j$, with the new mountain at position $j$ and of height $s_{j}=j+1$. 
    \item Heights are in weakly decreasing order in $S_n$, so mountains of the same height are always consecutive and thus mergeable. There are 0, 1 or 2 mountains of any height in the mountain list, and the height difference between consecutive mountains is 0, 1 or 2. 
\end{enumerate}
\end{lemma}

Index $j$ in point 3 can be expressed as $j=\nu_2(n+1)$, where $\nu_2(m)$ is the \emph{2-adic valuation} of~$m$: the highest integer $\nu$ such that $2^{\nu}$ divides~$m$. 
E.g., for $n=11$ items, there will be $t=\lfloor \log_2 (12) \rfloor=3$ mountains, and since $n+1$ in binary is $(1100)$, we ignore the leading bit and obtain the height sequence $S_{11}=(1,0,0)+(2,1,0)=(3,1,0)$, and the two mountains most recently merged were of height $j=\nu_2(12)=2$, as we see in Figure~\ref{fig:s10s11}. 
For comparison, in $\UMMR$~\cite{reyzin2016efficient} there are zero or one mountains of any height, and there is a height-$i$ mountain if and only if there is a ``one'' at the $i$-th position in the binary representation of $n$. 

\subsection{Resync operations and analysis}\label{s:UMMB-analysis}


The proofs of correctness, soundness and strength for U-MMB are identical to those in~\cite{reyzin2016efficient} for $\UMMR$, hence we skip them. 
Like the Merkle chain and $\UMMR$, $\UMMB$ is immutable: subtrees are never altered, as the structure only evolves by merging trees. 
Thus, all three Merkle structures are old-commitment compatible, because a membership proof $\pi_{x_i\in X}$, which contains the $i$-th leaf's ancestors' siblings, evolves as an append-only hash list, so it is a \emph{list prefix} of any future membership proof: $\pi_{x_i\in X_m} \prec \pi_{x_i\in X_n}$ for any $i\leq m\leq n$. 
For example, if the mountain in Figure~\ref{fig:membership} participates in a merge, the other mountain's peak hash would be appended to $\pi_{x_6\in X}$. 
Hence, if Bob sends proof $\pi_{x_i\in X_n}$ to Alice, she can easily extract from it the prefix $\pi_{x_i\in X_m}$ that is in sync with the commitment $\langle X_m \rangle$ she holds. 

An item's membership proof thus grows by one hash every time (the mountain containing) its leaf participates in a merge step, which, unlike in $\UMMR$, happens at most once per append. 
In fact, in any $k$-increment, any leaf participates in only $O(\log k)$ merge steps:

\begin{lemma}\label{lem:merge-steps} 
For any $k\geq 1$, within any $k$-increment there is at most one merge step in U-MMB where the height of the merging mountains is at least $\lceil\log_2 k\rceil$. 
Consequently, any leaf participates in at most $\lceil \log_2 k \rceil+1$ merge steps, and each of the $k$ new leaves ends in a mountain of height at most $\lceil \log_2 k \rceil+1$.
\end{lemma}

\begin{proof}
Fix a value $k\geq 1$. By Lemma~\ref{lem:Sn}, if we merge two mountains of height at least $\lceil\log_2 k \rceil$ during the merge step then $\nu_2(n+1)\geq \lceil \log_2 k \rceil$, or equivalently $2^{\lceil \log_2 k \rceil}$ divides $n+1$. 
The first claim follows from the fact that, clearly, this can happen at most once within any $k$ consecutive values of $n$. 
Next, consider a leaf that ends in a mountain of height $s$: 
if it participates in 0 or 1 merge steps there is nothing to show. 
Else, it participates in at least 2 merge steps where the height of the merging mountains is at least $s-2$. 
By the first claim, $s-2<\lceil \log_2 k \rceil$, or equivalently, $s\leq \lceil \log_2 k \rceil + 1$; this proves the second claim. And the third claim follows since each new leaf starts in a mountain of height zero, when it is appended.
\end{proof}

This result proves low-update frequency and incrementality of membership proof updates, and shows as well that recent items always have short proofs.  
We prove in Appendix~\ref{s:proofs} that the bound on the membership proof size of recent items can be slightly improved:

\begin{lemma}\label{lem:kth}
For any $1\leq k\leq n$, if $d:=\lfloor \log_2 k \rfloor+1$, then the $k$-th most recent leaf in U-MMB sits in one of the $d$ rightmost mountains, and its mountain is of height at most $d$.
\end{lemma}

Next, we focus on Alice's resync operation of updating her copy of the commitment from $\langle X_m \rangle$ to $\langle X_n \rangle$ after a $k$-increment. 
This commitment is a list of $O(\log n)$ peak hashes, but only $O(\log k)$ of these hashes need to be updated: 
First, we know from Lemma~\ref{lem:merge-steps} that there is at most one merge step where the merging peaks are of height at least $\lceil\log_2 k\rceil+1$, and we claim that these peaks must already appear in the old commitment $\langle X_m \rangle$, so Alice can recreate this merge locally, without any communication overhead. To prove the claim, we note that otherwise there must have been at least two merge steps with peaks of height at least $\lceil\log_2 k\rceil$ during the $k$-increment, which contradicts Lemma~\ref{lem:merge-steps}. 
Then, Alice simply retrieves from the bulletin board all peaks of height at most $\lceil\log_2 k\rceil+1$, and replaces these peaks in her local copy. By Lemma~\ref{lem:Sn}, there are at most $\lceil\log_2 k\rceil +2$ such peaks, and they are placed consecutively at the rightmost end of the peak list. 

\begin{lemma}\label{lem:differ}
    After an increment from $X_m$ to $X_n$ with $k=n-m$, in U-MMB the peak lists $\langle X_m \rangle$ and $\langle X_n \rangle$ differ in only $O(\log k)$ hashes.
\end{lemma}

\subsection{Recent-proof compatibility}\label{s:recent}

\begin{lemma}\label{lem:recent}
    Recent-proof compatibility implies low update frequency.
\end{lemma}

The proof of the result above is delayed to Appendix~\ref{s:proofs}. 
Finally, we prove that $\UMMB$ achieves the property of recent-proof compatibility, thus completing the proof of Theorem~\ref{thm:UMMB}, except for the claim on increment proofs, which is delayed to Section~\ref{s:increment}.  
We achieve this property by exploiting the fact that $\MMB$ evolves through lazy merges that can be predicted many appends in advance. 
We add the following rule to our scheme, which can be performed efficiently: 
\begin{itemize}
    \item When a membership proof $\pi_{x_i\in X}$ is generated or updated, if the corresponding leaf is in a mergeable pair of mountains, include the hash of the peak of the sibling mountain. 
\end{itemize}

In other words, we make the proof more ``future proof'' by emulating its state after these mountains are merged. 
This extended proof remains valid by the fact that U-MMB is old-commitment compatible, and the extension does not affect our previous analyses of the scheme, other than by increasing the bound on the proof size by one hash. 

\begin{lemma}\label{lem:rpc}
    U-MMB achieves recent-proof compatibility with parameter $c=1/5$.
\end{lemma}

\begin{proof} 
Consider an item $x_i$ whose proof $\pi_{x_i\in X_n}$ is updated at state $X_n$, $1\leq i\leq n$, when the $i$-th leaf is in a mountain $M$ of height $s$. 
Our strategy will be to find a lower bound $k$ on an increment that makes the proof invalid, and an upper bound $r$ on the recency ($n-i+1$) of the item, both as functions of $s$, and show their ratio observes $k/r > 1/5$. 
We assume that $s\geq 1$, as otherwise $i=n$, and the claim follows trivially by setting $k=r=1$.

Assume first that $M$ is not part of a mergeable pair. 
The proof will become invalid only after a)~another height-$s$ mountain is created, and then b)~it merges with $M$. 
As each event requires merging mountains of height at least $s-1$, by Lemma~\ref{lem:merge-steps} they cannot both happen within an increment of $2^{s-1}$ appends, so we get the bound $k=2^{s-1}+1$. 
On the other hand, at the future state where another height-$s$ mountain is created, we have by Lemma~\ref{lem:Sn} that the $s+1$ rightmost mountains must be of heights precisely $(s, s, s-2, s-3, \cdots, 1, 0)$, so they contain $5\cdot 2^{s-1} - 1$ leaves, including the $i$-th leaf. 
This proves that the recency of item $x_i$ at state $X_n$ is at most $r=5\cdot 2^{s-1} - 2$. 
We obtain the ratio $\frac{k}{r}\geq \frac{2^{s-1}+1}{5\cdot 2^{s-1} - 2}> \frac{1}{5}$.

Now assume that $M$ is in a mergeable pair, so proof $\pi_{x_i\in X_n}$ is extended. 
It will become invalid only after a)~$M$ merges with its sibling, and then b)~the resulting mountain merges again. 
As each event requires merging mountains of height at least $s$, by Lemma~\ref{lem:merge-steps} they cannot both happen within an increment of $2^{s}$ appends, so we get the bound $k=2^{s}+1$. 
Next, at the future state where $M$ merges with its sibling, we have by Lemma~\ref{lem:Sn} that the $s+1$ rightmost mountains must be of heights precisely $(s+1, s-1, s-2, \cdots, 1, 0)$, so they contain $3\cdot 2^s-1$ leaves, including the $i$-th leaf. 
Hence the recency of item $x_i$ at state $X_n$ is at most $r=3\cdot 2^s-2$. 
We obtain the ratio $\frac{k}{r}\geq \frac{2^{s}+1}{3\cdot 2^{s} - 2}> \frac{1}{3}> \frac{1}{5}$. 
This completes the proof.
\end{proof}


\section{F-MMB: a scheme with a constant-size commitment}\label{s:forward}

In the previous section we introduced a scheme that is succinct, incremental and optimally additive, as well as asynchronous, but requires keeping a commitment of logarithmic size. 
We propose further schemes that achieve a constant-size commitment and varying degrees of the first three properties, but lose the asynchrony properties. 
Instead, we consider a model where a scheme manager provides up-to-date proofs on demand to any participant. 


\begin{table}[htp]
\centering
\begin{tabular}{|l|c|c|c|c|}
\hline
                    & \textbf{MMR} & \textbf{F-MMR} & \textbf{F-MMB} & \textbf{MMB} \\
\hline
Succinctness                               & \cm       & \cm       & \cm       & \cm \\
Incrementality                             & \xm       & \xm       & partly    & \cm \\
CCC retrieve IP for $k$-increment           & $O(\log n)$  & $O(\log n)$  & $O(\log n)$  & $O(\log k)$ \\
CCC update MP after $k$-increment       & $O(\log n)$  & $O(\log n)$  & $O(\log n)$  & $O(\log k)$ \\
Worst-case MP size of $k$-th newest item$^*$    & $O(\log n)$  & $O(\log n)$  & $2\log_2 k$  & $2\log_2 k$ \\ 
Amortized MP size of $k$-th newest item$^*$ &$O(\log n)$ &$\frac{3}{2}\log_2 k$ &$2\log_2 k$ &$\frac{11}{8}\log_2 k$ \\
\hline
Optimal additivity                         & \xm       &on average &on average & \cm  \\
Worst-case append cost$^\dagger$ & $O(\log n)$  & $O(\log n)$  & $O(\log n)$  & $5$ \\
Amortized append cost$^\dagger$  & $O(\log n)$  & $2$       & $3$       & $4$ \\
\hline
\end{tabular}
\caption{Comparison of variants of $\MMR$~\cite{MMR} and $\MMB$, assuming $1\leq k\leq n=|X|$. 
MP and IP stand for membership and increment proof, respectively. 
$^\dagger$Cost measured in hash computations per append, excluding the hash of the new leaf. 
$^*$Number of hashes in proof $\pi_{x_{n-k+1}\in X_n}$, with a hidden additive $O(1)$ term, as $k$ is fixed and $n$ is a variable; see Appendix~\ref{s:amortized} for details. 
} 
\label{tab:blockchain}
\end{table}

For this model we propose two scheme variants, $\FMMB$ and $\MMB$, presented in this and the next section, respectively. 
For the sake of completeness, in Table~\ref{tab:blockchain} we compare them against two possible variants of $\MMR$~\cite{MMR} that we describe in Section~\ref{s:bagging}. 
We include amortized complexities as well as worst-case complexities. 
All four schemes are succinct and maintain a constant-size commitment, but they reach other properties to varying degrees. 
In particular, $\FMMB$ provides short membership proofs for recent items, and optimal additivity on average, while $\MMB$ achieves full incrementality and optimal additivity.

We present $\FMMB$ mostly as a pedagogical stepping stone in our presentation, as its structure is considerably simpler than that of $\MMB$. 
However, $\FMMB$ may be ideal for applications where participants neither store membership proofs nor require increment proofs.  

\subsection{Bagging the peaks of a mountain range}\label{s:bagging}

We follow Peter Todd's terminology used in $\MMR$~\cite{MMR}. 
That structure consists of two layers:  
the bottom layer is the list of mountains in $\UMMR$~\cite{reyzin2016efficient}, 
while in the top layer peaks are linked together with a Merkle chain. 
Todd refers to this linking as \emph{peak bagging}, and to the resulting structure as a \emph{mountain range}.%
\footnote{Among hikers, \emph{peak bagging} is an activity that consists of attempting to reach all summits of a list of mountains, which may be of high popularity or significance in a geographical area.} 
This is why we refer to the scheme in~\cite{reyzin2016efficient} as ``unbagged'' $\MMR$ ($\UMMR$). 
In $\MMR$, peaks are bagged in a \emph{backward} fashion, from right to left, which helps keep the structure well balanced; see Figure~\ref{fig:bagging}. 
Let the ``forward-bagged'' $\MMR$ ($\FMMR$) be the variant where peaks are bagged in the opposite way.

\begin{figure}[htp]
\centering
\includegraphics[width=0.8\textwidth]{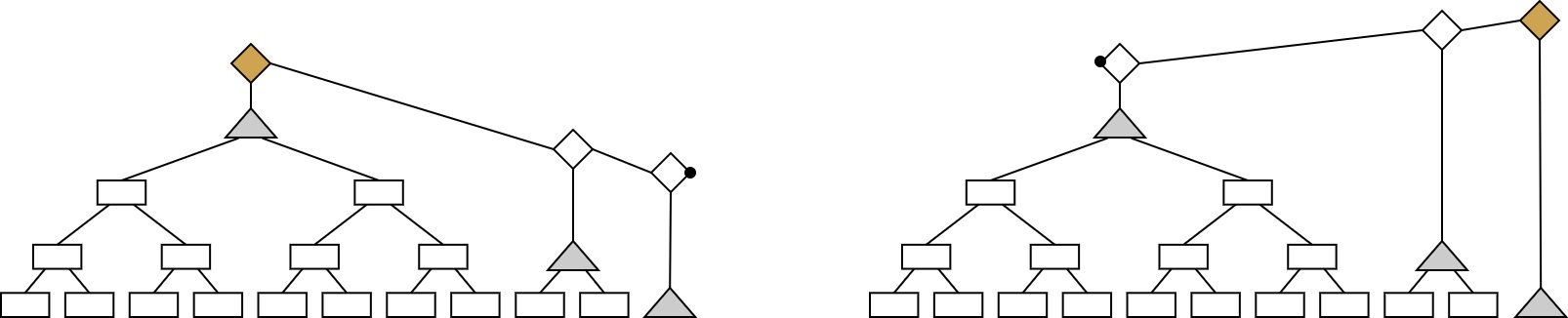}
\caption{For $n=11$, $\UMMB$ and $\UMMR$ coincidentally produce the same height sequence $S_{11}=(3,1,0)$. 
We backward-bag these mountains on the left and forward-bag them on the right. 
Hence, the structure on the left is an $\MMR$, while that on the right is an $\FMMR$ and an $\FMMB$. 
}
\label{fig:bagging}
\end{figure}

The \emph{forward-bagged} $\MMB$ ($\FMMB$) has a similar two-layer structure: 
the bottom layer consists of the list of mountains in $\UMMB$, described in the previous section, while in the top layer we forward-bag the peaks. 
We refer to the top-layer bagging nodes as \emph{range nodes}, and represent them visually with rhombi.  
To avoid visual clutter, moving forward we hide internal (i.e., non-peak) mountain nodes, so that each mountain is represented only by its peak, which is a triangle labeled with its height; see Figure~\ref{fig:FMMB}. 
Also, we place all peaks on one horizontal level and all range nodes on another horizontal level, where a forward bagging is always assumed, so the root is the rightmost range node.

\begin{figure}[htp]
\centering
\includegraphics[width=0.45\textwidth]{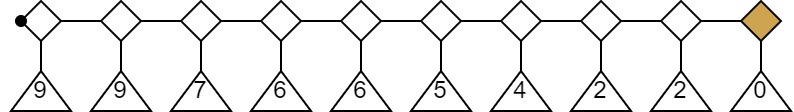}
\caption{Simplified representation of F-MMB for $n=1337$. 
Internal mountain nodes are hidden, peaks are triangles labeled with their height, range nodes are rhombi, and the root is in brown.}
\label{fig:FMMB}
\end{figure}

\subsection{Append operation with forward bagging}\label{s:app-f}

As there are $O(\log n)$ mountains in $\MMR$~\cite{MMR}, its backward bagging takes logarithmic time and needs to be fully recomputed after each append. 
The $\FMMR$ variant is just as expensive if performed from scratch, but has the advantage that most range nodes do not need to be recomputed after an append, as one only needs to update the nodes above and to the right of the mountains affected by merges. 
In fact, it can be checked that the $\FMMR$ append time remains worst-case logarithmic, but reduces to constant on average.%
\footnote{The amortized number of hash computations per append in $\FMMR$ is $2$: the amortized number of merge steps is $1$, and then a single bagging node is (re-)computed above the resulting rightmost peak.}

Similarly, in $\FMMB$ we perform the same append operation as in $\UMMB$, with the addition of a partial forward rebagging that updates a minimum number of range nodes. 
Hence, in $\FMMB$ the \textbf{append operation}, performed by the scheme manager, is: 
\begin{enumerate}
    \item \emph{Add step:} add the new leaf as a height-0 mountain to the right end of the mountain list; 
    \item \emph{Merge step:} if there are mergeable mountains, merge the rightmost pair, so that the new mountain is one unit higher and replaces the pair in-place in the mountain list; and 
    \item \emph{Bag step:} if there is a merge step, update all range nodes above and to the right of the merge peak; otherwise, create a new range node above the new leaf.
\end{enumerate}

\begin{figure}[htp]
\centering
\includegraphics[width=0.75\textwidth]{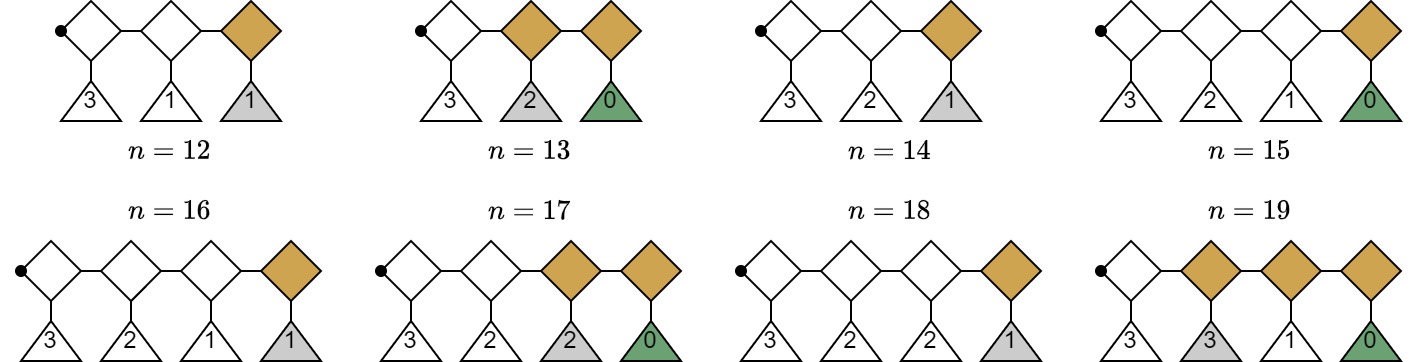}
\caption{Evolution of the $\FMMB$ structure during a sequence of item appends, starting from that on the right-hand side of Figure~\ref{fig:bagging}. 
We color all nodes that are created or updated during each append, including the new leaf (in green), the merge peak (in gray) and new range nodes (in brown).}
\label{fig:FMMBs}
\end{figure}

\begin{lemma}\label{lem:hash-F} 
The number of hash computations needed during the $\FMMB$ append operation is at most $\lfloor \log_2(n+1) \rfloor+1$, with an amortized value of $3$.
\end{lemma}

We prove Lemma~\ref{lem:hash-F} in Appendix~\ref{s:proofs}, and only mention here that the amortized value corresponds to one hash computation in the bottom layer (see Lemma~\ref{lem:hash-U}) plus an average of only two updates in the top layer with partial rebagging. 
The manager keeps all peaks and range nodes in cache in order to perform the append, hence the logarithmic memory requirement. 
We thus obtain the following result, whose formal proof is skipped.

\begin{lemma}\label{lem:append-F}
    To perform the F-MMB append operation, the scheme manager requires $O(\log n)$ memory, and time that is $O(\log n)$ in the worst case and $O(1)$ on average. 
\end{lemma}

\subsection{Resync operations and analysis}

While a backward bagging would result in a more balanced overall structure, we embrace the slight imbalance of forward bagging, as it provides the following useful property.

\begin{lemma}\label{lem:f-distance} 
For any $1\leq k\leq n$, in F-MMB the $k$-th most recently appended item has a membership proof consisting of at most $2\lfloor \log_2 k \rfloor +2 =O(\log k)$ hashes. 
\end{lemma}
\begin{proof} 
By Lemma~\ref{lem:kth}, if $d:=\lfloor \log_2 k \rfloor +1$ then the $k$-th most recent leaf is in one of the $d$ rightmost mountains, whose height is at most $d$. 
This means that the distance from the leaf to the mountain peak is bounded by $d$, and the distance from this peak to the range root is bounded by $d$ again, so the overall leaf-to-root distance is at most $2d$.
\end{proof}

Hence, recent items have short proofs, regardless of the list size $n$. 
The $\FMMB$ scheme is thus partly incremental, as stated in Table~\ref{tab:blockchain}, and as such already provides most of the protocol improvements highlighted in Table~\ref{tab:apps}, when compared to $\MMR$. 
To conclude the section, we observe that all other resync operations can be performed with little communication complexity, but may require more computation. 

\begin{lemma} 
In $\FMMB$, a participant resynchronizing after a $k$-increment needs $O(\log n)$ computational and $O(\log k)$ communication complexity to perform any resync operation.
\end{lemma}

To see why this statement holds, notice that a participant can store in local memory all mountain peaks, execute any $\UMMB$ resync operation, and then build the bagging nodes locally in $O(\log n)$ time but without any additional communication overhead. 
Hence, the statement follows from Theorem~\ref{thm:UMMB}. 
Of course, this requires the scheme manager to answer queries about updates to the peak list. 
For a participant not storing mountain peaks, both the communication and computational complexities will be $O(\log n)$.

\section{The Merkle Mountain Belt}\label{s:double}

Consider the different steps of the append operation: lazy merging reduces the merging cost from logarithmic to constant, and similarly forward bagging reduces the bagging cost from logarithmic to constant \emph{on average}. 
Yet, when the merge takes place far from the right end of the mountain list, we need to update many range nodes. 
In turn, this substantially changes the membership proofs of most items, making resync operations non-incremental.  
In this section we propose a novel bagging procedure that circumvents this issue.  

As a first approach, consider partitioning the U-MMB mountain list into two sublists of roughly equal size, and forward-bagging each one to obtain two separate mountain ranges. 
This way, when a merge takes place, we only need to update range nodes within the affected range, leaving the other range unaffected.  
This trick effectively halves the worst-case bound on the number of range node updates triggered by a merge. 

The next idea, of course, is to partition the mountains into even more ranges. 
But rather than fixing a size for these ranges, we define their boundaries in terms of the position of the mergeable pairs of mountains:  
our goal is to have each mergeable pair placed at the very right end of a range. 
This way, wherever a merge occurs, only the range node directly above the merge peak needs to be updated, as it is the root of its range. 

However, we now need an additional, third layer, where all mountain ranges are forward-bagged once again into a \emph{mountain belt}.%
\footnote{In geology, a mountain belt is a group of mountain ranges with similarity in structure and a common cause, usually an orogeny. 
E.g., in Ecuador, the Andes mountain belt consists of two parallel ranges, the Cordillera Occidental and the Cordillera Oriental.} 
This gives rise to the Merkle Mountain Belt (MMB), a succinct, incremental, and optimally additive scheme with an $O(1)$ sized commitment. 

\subsection{Append operation with double bagging}\label{s:MMB-append}

We introduce the $\MMB$ structure with the help of a visual example in Figure~\ref{fig:MMB}; compare it against the $\FMMB$ structure in Figure~\ref{fig:FMMB} for the same value of $n$. 
As before, we hide internal mountain nodes, represent peaks with triangles and place them on one horizontal line, represent range nodes with rhombi and place them on a second line, and finally we refer to the top-layer bagging nodes as \emph{belt nodes}, represent them with circles, and place them on a third line. 
A range with $i$ mountains requires $i$ range nodes, and similarly $j$ ranges require $j$ belt nodes, giving a total of at most $2t$ bagging nodes for a list of $t=\lfloor \log_2 (n+1) \rfloor$ mountains. 
All peaks and bagging nodes are kept in memory by the manager, hence the logarithmic memory requirement.

\begin{figure}[htp]
\centering
\includegraphics[width=0.55\textwidth]{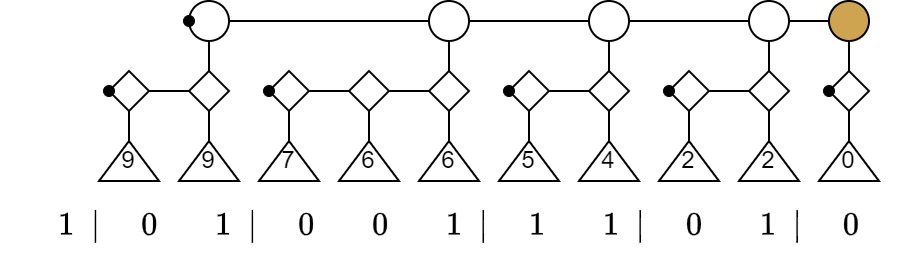}
\caption{MMB structure for same value $n=1337$ as in Figure~\ref{fig:FMMB}, and binary representation of $n+1$ with encoded range splits.  
Peaks are represented with triangles, range nodes with rhombi and belt nodes with circles. The root is in brown and its hash is the commitment.}
\label{fig:MMB}
\end{figure}

The MMB structure is fully characterized once we establish how the U-MMB mountain list is partitioned. 
We do so with the following \textbf{range split conditions}: consecutive mountains $(M, M')$ belong to different ranges if and only if
\begin{enumerate}
    \item their height difference is 2, or
    \item $M$ forms a mergeable pair with the mountain to its left.
\end{enumerate}

An equivalent characterization, which will be useful for the scheme analysis, can be established in terms of the binary representation of $n+1$: a range split occurs wherever there is a bit sequence $(1|0)$ or $(01|)$ in it, where the vertical bar indicates the location of the corresponding split; see Figures \ref{fig:MMB} and \ref{fig:MMBs}. 
The equivalence follows easily from Lemma~\ref{lem:Sn}.

While the second split condition guarantees that each mergeable pair sits at the right end of a range, the first condition ensures that, after a merge, the new merge peak \emph{remains} at the right end of its range, and does not join the range to its right (which would trigger many range node updates).  
However, if the merge peak sits alone in its range, it might join the range to its left. 
Similarly, the new leaf joins the rightmost range whenever possible.
Hence, these rules minimize the number of range node updates needed, while also not creating more splits than necessary. 
We present a series of appends in Figure~\ref{fig:MMBs} for illustration purposes.

\begin{figure}[htp]
\centering
\includegraphics[width=0.85\textwidth]{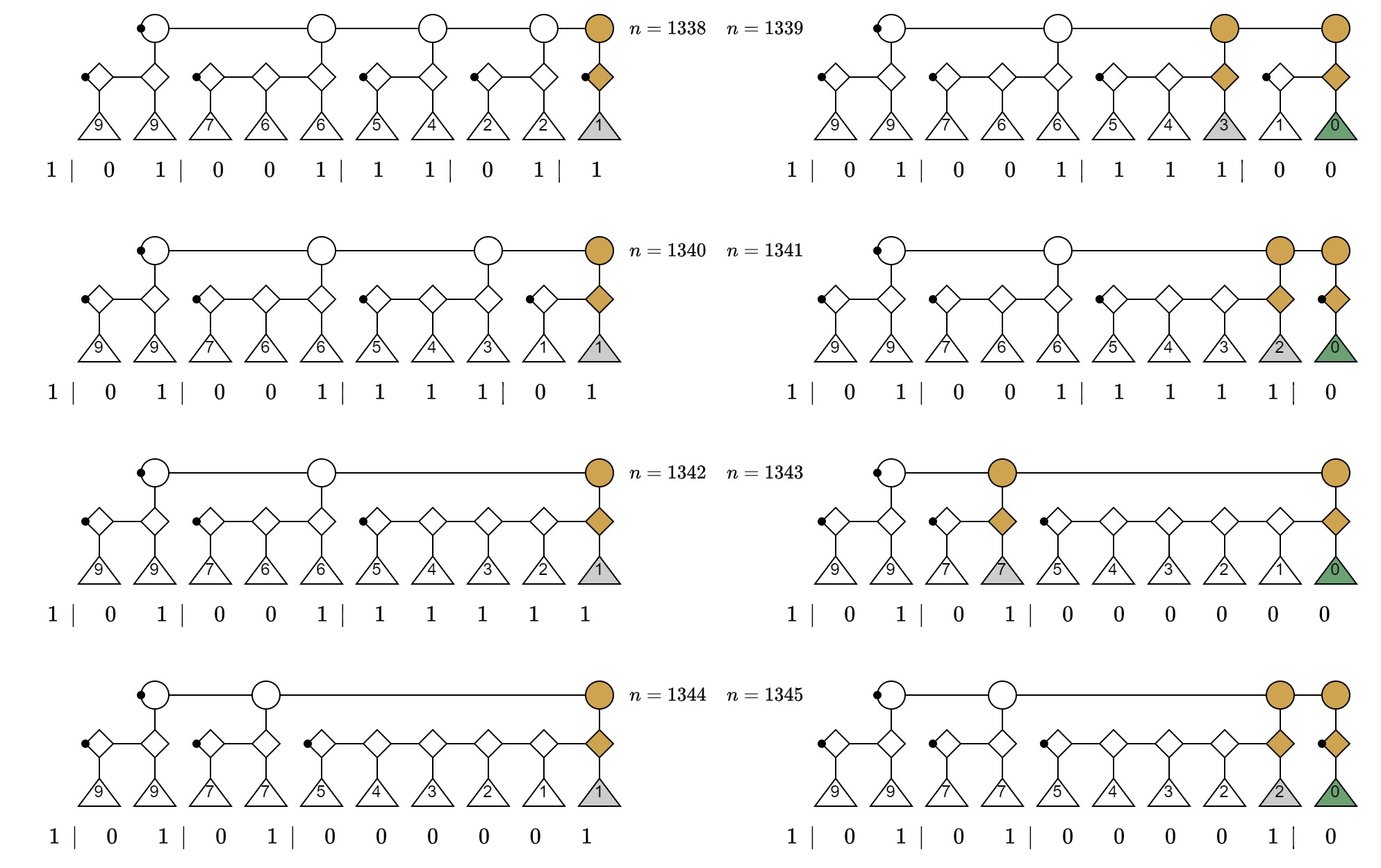}
\caption{Evolution of the MMB structure during a sequence of item appends, starting from the structure in Figure~\ref{fig:MMB}. 
We color all nodes that are created or updated during each append, including the new leaf (in green), the merge peak (in gray), and new range and belt nodes (in brown).}
\label{fig:MMBs}
\end{figure}

During an append, we execute a partial rebagging that preserves as many nodes as possible. 
Naturally, we need to update the range and belt nodes directly above the merge peak, as well as the range and belt nodes directly above the new leaf, but we claim that \emph{no other bagging nodes need updating}. 
This is clear from our previous discussion for the range nodes. 
But how do we avoid a long sequence of updates among belt nodes?
The answer lies in the next observation, which states that although the merge peak may sit arbitrarily far from the right end of the mountain list, it is always conveniently close to the belt root.

\begin{lemma}\label{lem:close}
If an append operation in MMB has a merge step, the new merge peak sits at the right end of its range. Moreover, it sits in either the rightmost or second rightmost range.
\end{lemma}

\begin{proof}
Recall from Lemma~\ref{lem:Sn} that the new merge peak sits at position $\nu_{2}(n+1)$, which is the position of the rightmost ``one'' in the binary representation of $n+1$, so there are only zeroes to its right. 
As the range split conditions produce a split between a one and a zero, but no splits within a sequence of zeros, we conclude that this merge peak sits at the right end of a range, while all mountains to its right (if any) form another single range.
\end{proof}

\begin{lemma}\label{lem:hash-d} 
The number of hash computations needed during the $\MMB$ append operation is at most $5$, with an amortized value of $4$.
\end{lemma}
 
A proof of this result is included in Appendix~\ref{s:proofs},  
and we remark that the $\MMB$ append requires only one additional hash computation on average relative to $\FMMB$; see Lemma~\ref{lem:hash-F}. 

\begin{lemma}
    The $\MMB$ append operation requires $O(\log n)$ memory and $O(1)$ time.
\end{lemma}

\subsection{Resync operations and analysis}\label{MMB-analysis}

\begin{lemma}\label{lem:d-distance} 
For any $1\leq k\leq n$, in MMB the $k$-th most recent item has a membership proof consisting of at most $2\lfloor \log_2 k \rfloor +3 =O(\log k)$ hashes. 
\end{lemma}

\begin{proof}
    The result follows from Lemma~\ref{lem:f-distance} and the fact that, for any fixed $n$, the depth of any leaf increases by at most one unit in MMB relative to F-MMB. 
    Indeed, range splits only shorten this depth, so the largest increase relative to F-MMB takes place in the rightmost range, where any leaf has a single additional ancestor, namely a belt node.
\end{proof}

Hence, we see that double bagging still provides short proofs to recent items. 
Now, proving that all resync operations observe the incrementality property is trickier, and requires us to first introduce some definitions and a technical result. 
Relative to a $k$-increment from $X_m$ to $X_n$, let the \emph{merge range} be the range containing the highest merge peak created, 
i.e., the leftmost range affected during the increment, and let the \emph{leaf range} be the range containing the $(m+1)$-th leaf. 
We prove that these two ranges are never far apart. 

\begin{lemma}\label{lem:merge-leaf}
After any $k$-increment in MMB, the merge and leaf ranges are either equal or adjacent. 
Moreover, the leaf range may be arbitrarily large but only $O(\log k)$ rightmost peaks in it are new, and there are $O(\log k)$ further peaks (and ranges) to the right of the leaf range.  
\end{lemma} 

\begin{proof} 
Consider the state during the $k$-increment at which the highest merge peak is created. 
By Lemma~\ref{lem:close}, at this point the merge range is either the rightmost or second rightmost range, with the leaf range to its right, so these ranges are either equal or adjacent.  
Furthermore, there are no mergeable pairs to the right of the merge peak, so in particular \emph{there is no mergeable pair strictly in between the highest merge peak and the $(m+1)$-th leaf}. 

We claim that this last property continues to hold in later appends. 
The proof is by induction: mountain heights only change via merges, and any further merge (taking place to the right of the highest merge peak) must either affect the $(m+1)$-th leaf or take place to its right by induction hypothesis. 
Hence, the merge and leaf ranges do not drift apart, because new range splits only take place to the right of new merges. 
This proves the first statement. 

The second and third statements follow from Lemma~\ref{lem:merge-steps}: 
besides the highest merge peak, all other merges are of height $O(\log k)$, so at the end of the $k$-increment there can only be $O(\log k)$ new peaks, placed either at the right end of the leaf range, or to its right. 
\end{proof}

For instance, Figure~\ref{fig:prefix-D} shows the result of a $k$-increment from $m=91107$ to $n=91145$, $k=n-m=38$, where the highest merge peak is of height $11$. 
The merge and leaf ranges are adjacent, and the leaf range is large but contains only $3$ new peaks, with $3$ additional peaks --~and $2$ ranges~-- to its right. 
In that figure we label each mountain node with its height, defined as its distance to any descendant leaf. 

\begin{figure}[htp]
\centering
\includegraphics[width=0.7\textwidth]{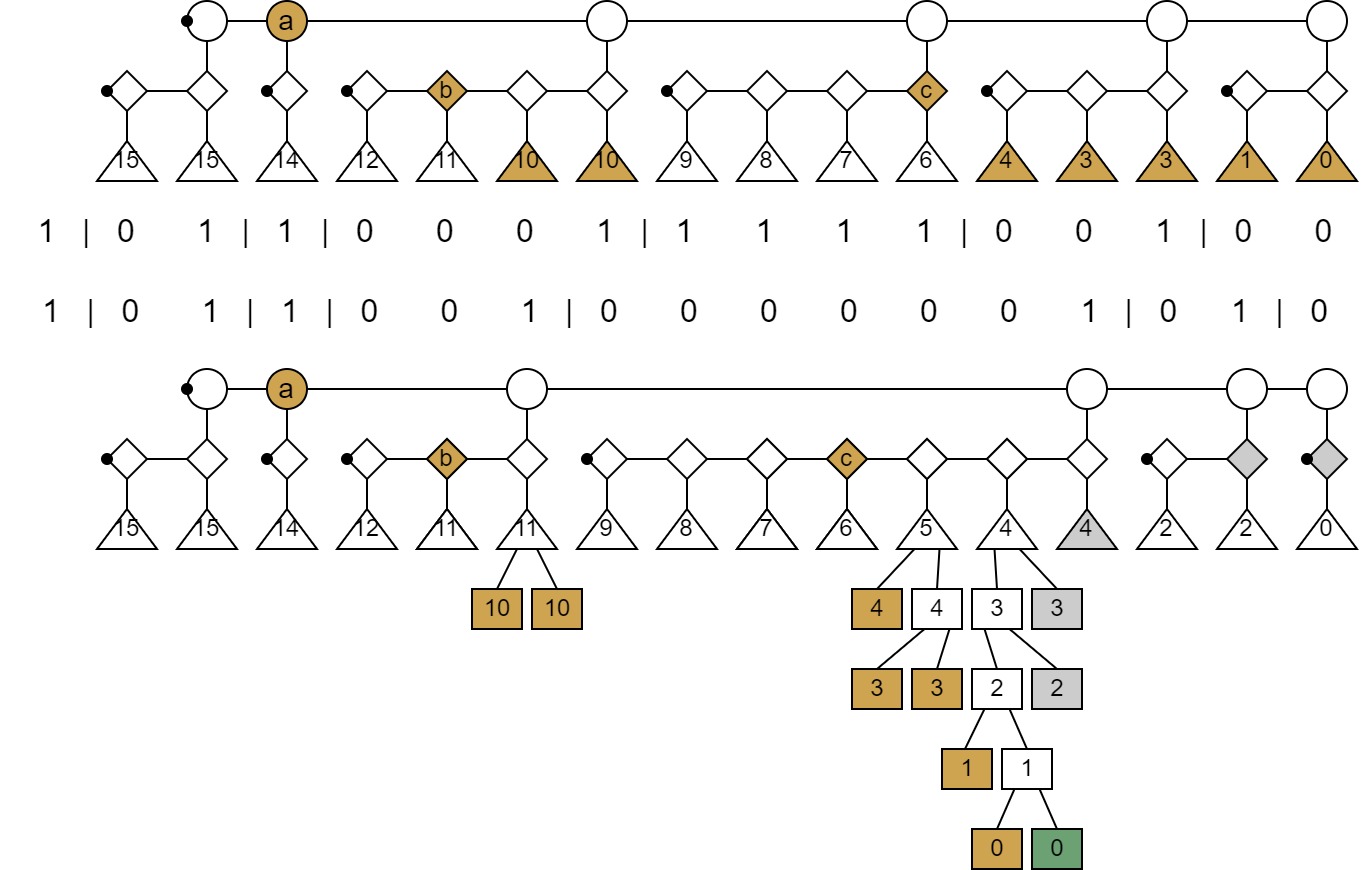}
\caption{MMB structure for $m=91107$ (above) and $n=91145$ (below), with the binary representations of $m+1$ and $n+1$, and mountain nodes labeled with their height. 
Non-white nodes form set~$C_n$, brown nodes form its subset $C_m$ (see Section~\ref{s:increment}), and the $(m+1)$-th leaf is in green.}
\label{fig:prefix-D}
\end{figure}

We now show that we can update any membership proof efficiently, which completes the proof of Theorem~\ref{thm:MMB}, except for the claim on increment proofs, which is delayed to Section~\ref{s:increment}. 
In what follows, let \textbf{a} be the belt node that commits to all unchanged ranges (i.e., those to the left of the merge range), and let \textbf{b} and \textbf{c} be the range nodes that commit to all unchanged mountains in the merge range and leaf range, respectively; again see Figure~\ref{fig:prefix-D}.

\begin{lemma}
After a $k$-increment, in MMB the path from any leaf to the belt root only changes by $O(\log k)$ topmost nodes; 
hence, any membership proof changes by $O(\log k)$ hashes. 
\end{lemma}

\begin{proof} 
Let $h_i=H(x_i)$ be the leaf in consideration. 
We consider 3 cases: 
\begin{enumerate} 
    \item Bagging nodes \textbf{a}, \textbf{b} and \textbf{c} are all unchanged (by definition) and at distance $O(\log k)$ to the root by Lemma~\ref{lem:merge-leaf}. So, if $h_i$ is a descendant of any of these nodes, the claim follows. 
    \item Similarly, the highest merge peak is at distance $O(\log k)$ to the root by Lemma~\ref{lem:merge-leaf}, and is the first ancestor to change for all its descendant leaves, so the claim follows for them. 
    \item The only remaining case is for $h_i$ to be in the leaf range, and in a mountain that has participated in merges. By Lemma~\ref{lem:merge-steps}, its mountain must be of height $O(\log k)$, so together with Lemma~\ref{lem:merge-leaf} we conclude that its full leaf-to-root distance is $O(\log k)$. 
\end{enumerate} 
\end{proof}

\section{The increment proof}\label{s:increment}

In this section we formalize the increment proof for any append-only commitment scheme based on a Merkle structure. 
Then, we describe it in detail for both U-MMB and MMB.

Given a Merkle structure $\mathcal{M}$ for list $X$, let $h_i:=H(x_i)$ be its $i$-th leaf, and let $[h_i, h_j]:=\{h_i, h_{i+1}, \cdots, h_j\}$ be a range of consecutive leaves, for any $1\leq i\leq j\leq n=|X|$. 
A node set $C$ in $\mathcal{M}$ is \emph{independent} if no two nodes in it have any ancestor-descendant relationship. 
An \emph{opening step} in set $C$ is replacing a node by its children, and a \emph{closing step} is replacing two siblings by its parent. 
Following notation from~\cite{meiklejohn2020think}, we say that an independent set $C$ \emph{covers} $[h_i, h_j]$ if this range is exactly the set of leaves that descend from nodes in~$C$. 
Commitment $\langle X_n \rangle$ trivially covers $[h_1, h_n]$, and we can verify that an independent set $C$ covers $[h_1, h_n]$ by deriving $\langle X_n \rangle$ from $C$ via a sequence of closing steps. 
Instructions on these steps should be provided to make this verification efficient.%
\footnote{In this section we ignore the communication and computational complexity (CCC) of these instructions, as the overall CCC linked to an increment proof is typically dominated by the number of hashes in it.} 
By the properties of a cryptographic hash function, it is computationally hard to derive $\langle X_n \rangle$ from a set $C$ that does not cover $[h_1, h_n]$. 

We now provide an initial definition of an increment proof, as well as an alternative definition of a membership proof, that will come in handy:

\begin{itemize} 
    \item An \emph{extended} increment proof $\pi'_{X_m\prec X_n}$ is an independent set $C_n$ containing a subset $C_m$, such that $C_m$ and $C_n$ provably cover $[h_1, h_m]$ and $[h_1, h_n]$, respectively.
    \item For an item $x_i\in X_n$, an \emph{extended} membership proof $\pi'_{x_i\in X_n}$ is an independent set $C$ that contains $h_i$ and provably covers $[h_1, h_n]$.  In particular, $C\setminus h_i$ can be partitioned into two subsets, that cover $[h_{1}, h_{i-1}]$ and $[h_{i+1}, h_{n}]$, respectively. 
\end{itemize}

To see why the first definition makes sense, notice that if any of the first $m$ leaves was modified between the publications of $\langle X_m \rangle$ and $\langle X_n \rangle$, then either $\langle X_m \rangle$ cannot be derived from $C_m$, $\langle X_n \rangle$ cannot be derived from $C_n$, or $C_m$ cannot be contained in $C_n$. 
Notice as well that the extended membership proof $\pi'_{x_{m+1}\in X_n}$ of the $(m+1)$-th item is also an increment proof $\pi'_{X_m\prec X_n}$, provided that instructions are given to efficiently derive from it both $\langle X_m \rangle$ and $\langle X_n \rangle$. 
This will be our general strategy to build increment proofs. 

For instance, in the Merkle tree in Figure~\ref{fig:membership}, set $C_8:=\{h_{14}, h_{5}, h_6, h_{78}\}$ contains $h_6$ and covers $[h_1, h_8]$, so it is a membership proof $\pi'_{x_6\in X_8}$. 
It also contains the subset $C_5:=\{h_{14}, h_{5}\}$, which covers $[h_1, h_5]$, so $C_8$ is an increment proof $\pi'_{X_5\prec X_8}$.

In the case of $\UMMB$, we start from the extended membership proof $\pi'_{x_{m+1}\in X_n}$ of the $(m+1)$-th item, and perform additional opening steps to include all nodes that were mountain peaks before the $k$-increment; see Figure~\ref{fig:prefix}. 
Notice that this proof is of size $O(\log n)$ and, in particular, contains \emph{all} peaks in $\langle X_m \rangle$; i.e., $C_m = \langle X_m \rangle$. 

\begin{figure}[htp]
\centering
\includegraphics[width=0.7\textwidth]{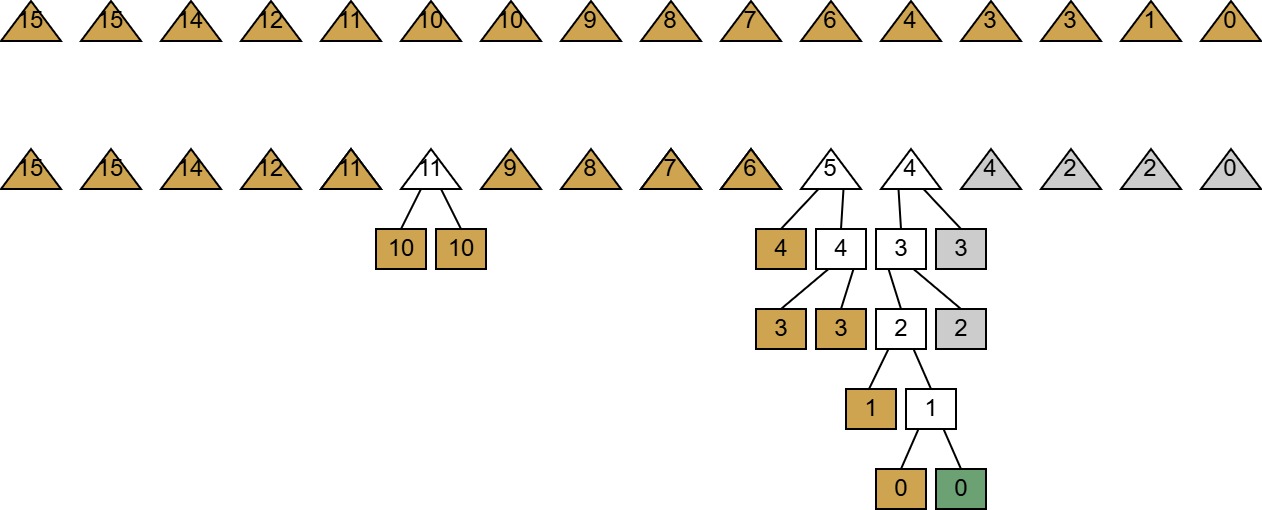}
\caption{U-MMB structures for the same values of $m$ and $n$ as in Figure~\ref{fig:prefix-D}. 
Non-white nodes form set $C_n=\pi'_{X_m\prec X_n}$, brown nodes form its subset $C_m=\langle X_m\rangle$, the $(m+1)$-st leaf is in green, and finally green and  gray nodes form the restricted increment proof $\pi_{X_m\prec X_n}= C_n \setminus \langle X_m\rangle$.}
\label{fig:prefix}
\end{figure}

However, for any Merkle structure, we define \emph{restricted} increment and membership proofs from extended ones by removing some redundant information:  
\begin{itemize} 
    \item $\pi_{X_m\prec X_n}:=\pi'_{X_m\prec X_n}\setminus \langle X_m \rangle$, and 
    \item $\pi_{x_i\in X_n}:=\pi'_{x_i\in X_n} \setminus \langle X_n \rangle$. 
\end{itemize}  

These are sensible restrictions for a verifier who knows the old commitment $\langle X_m \rangle$ prior to requesting an increment proof, and knows the current commitment $\langle X_n \rangle$ prior to requesting a membership proof. 
Notice that these restrictions make no difference for single-tree structures (such as MMB), and that the restricted membership proof corresponds to the standard one. 

\begin{lemma}
The increment proof $\pi_{X_m\prec X_n}$ in U-MMB is of size $O(\log k)$, where $k=n-m$.    
\end{lemma}

\begin{proof}
This proof contains 
i)~nodes from the restricted membership proof $\pi_{x_{m+1}\in X_n}$, of size $O(\log k)$ by Lemma~\ref{lem:kth}, and 
ii)~nodes from set $\langle X_n \rangle \setminus \langle X_m \rangle$, of size $O(\log k)$ by Lemma~\ref{lem:differ}. 
\end{proof}

In U-MMB, proof $\pi_{X_m\prec X_n}$ simply corresponds to a set covering $[h_{m+1}, h_{n}]$, and its verification consists of deriving $\langle X_n \rangle$ from $(\langle X_m \rangle \cup \pi_{X_m\prec X_n})$ via closing steps. 
In fact, this increment-proof verification could be the standard means by which a verifier updates their local copy of the commitment, from $\langle X_m \rangle$ to $\langle X_n \rangle$, after a $k$-increment.

In MMB we follow a similar procedure: we start from the membership proof of the $(m+1)$-th item, which is a set $C_n$ covering $[h_1, h_n]$ with a subset $C_m$ covering $[h_1, h_m]$, and we perform additional opening steps in $C_m$ to include all nodes that were peaks before the $k$-increment but not thereafter, i.e., nodes in $\langle X_m\rangle_{\UMMB} \setminus \langle X_n \rangle_{\UMMB}$; see the example in Figure~\ref{fig:prefix-D}. 
A verifier simply needs to derive from $C_m$ and $C_n$ the roots $\langle X_m\rangle_{\MMB}$ and $\langle X_n \rangle_{\MMB}$, respectively, via closing steps. 

\begin{lemma}
The increment proof $\pi_{X_m\prec X_n}$ in MMB is of size $O(\log k)$, where $k=n-m$.     
\end{lemma}
\begin{proof}
This proof contains: 
i)~nodes from the membership proof $\pi_{x_{m+1}\in X_n}$, which is of size $O(\log k)$ by Lemma~\ref{lem:d-distance}; 
ii)~nodes from the set difference $\langle X_m\rangle_{\UMMB} \setminus \langle X_n \rangle_{\UMMB}$, which is of size $O(\log k)$ by Lemma~\ref{lem:differ}; and 
iii)~the bagging nodes \textbf{a}, \textbf{b} and \textbf{c}, defined in Section~\ref{MMB-analysis}. 
\end{proof}

We conclude that, for both $\UMMB$ and $\MMB$, the restricted increment proof $\pi_{X_m\prec X_n}$ is of size $O(\log k)$ and allows a verifier to check the relation $X_m\prec X_n$, assuming they know commitment $\langle X_m \rangle$. 
This completes the proof of the incrementality property in our schemes.


\section{Conclusions}\label{s:conc}

These are the main advantages of $\MMB$ relative to $\MMR$:  
1)~each item append takes constant time; 
2)~after a $k$-increment offline, a light client can resynchronize with a cost polylogarithmic in $k$, instead of $n$;  
3)~$\UMMB$ (and a variant of $\MMB$ presented in Appendix~\ref{s:apps}) offers light clients the new asynchrony property of recent-proof compatibility; and 
4)~in many applications where light clients exhibit a recency bias in their queries, the size of the proofs they handle ~-- and their transaction fees~-- become constant over time, in expectation; see Table~\ref{tab:apps}. 

Among blockchain applications with a commitment to an ever-growing append-only list~$X$, the authors have yet to find an instance where $\MMB$ is \textbf{not} asymptotically superior to both $\MMR$ and the Merkle chain. 
See Appendix~\ref{s:apps} for an overview of the improvements brought by $\MMB$ to some of these applications. 
Hence, we propose $\MMB$ as the \emph{canonical} Merkle structure for such applications; most notably, for committing to block headers.

\vspace{5mm}

\textbf{Acknowledgments.} 
This project started when all authors were affiliated with Web 3.0 Technologies Foundation, and continued while Robert Hambrock was affiliated with Parity Technologies. 
The work of Alfonso Cevallos and Robert Hambrock was also greatly supported by a grant from OpenGov --~a community-governed funding mechanism within the Polkadot network~\cite{burdges2020overview}~-- approved via Referendum $\#1667$ in August 2025. 
The authors gratefully acknowledge this support from the Polkadot community, which made it possible to complete the present work. 
\appendix

\bibliography{zreferences}

\section{Related work and potential applications}\label{s:apps}

In this section we briefly describe some protocols --~theoretical and implemented ones, in blockchain and elsewhere~-- that currently use $\MMR$~\cite{MMR} or $\UMMR$~\cite{reyzin2016efficient} as a commitment scheme, and we highlight the improvements they would obtain by using one of our proposed schemes instead. 
We also describe further possible variants of $\MMB$. 

In particular, in recent years many blockchain communities have discussed transitioning from a Merkle chain to an $\MMR$ to commit to block headers. 
This can be accomplished via a hard fork, a soft fork, or a velvet fork; see~\cite{moshrefi2021lightsync} for a discussion on these options. 
Such a transition has been proposed for Ethereum Classic~\cite{perazzo2024smartfly}, and adopted by Grin, Beam, Harmony~\cite{lan2021horizon}, Ethereum (with the Herodotus protocol~\cite{Herodotus}), Zcash~\cite{hopwoodzcash} (with the adoption of FlyClient~\cite{bunz2020flyclient} in its Heartwood fork), and Polkadot~\cite{burdges2020overview} (in its BEEFY protocol~\cite{bhatt2025trustless}). 
This adoption of $\MMR$ has enabled faster LCPs and cross-chain bridges. 
These could be made even faster with $\MMB$. 

\vspace{3mm}

\textbf{FlyClient.} 
Consider a light client Alice who wants to synchronize to a PoW blockchain network, but might be in the presence of adversarial forks, which contain invalid blocks. 
FlyClient~\cite{bunz2020flyclient} proposes committing to the list $X$ of block headers with an MMR, and requires each header $x_i$ to store the commitment $\langle X_{i-1} \rangle$ to all previous headers. 
Authors prove that, after downloading the current commitment $\langle X_{n} \rangle$, Alice can detect an adversarial fork with overwhelming probability by requesting the headers $x_i$ of $O(\log n)$ randomly sampled blocks, checking their PoW difficulty level, and downloading and verifying their corresponding proofs $\pi_{x_i\in X_n}$ and $\pi_{X_{i-1}\prec X_n}$. 
As each of these $O(\log n)$ proofs is of size $O(\log n)$ with MMR, they obtain a total communication and computational complexity (CCC) of $O(\log^2 n)$. 

Using $\MMB$ would decrease the update time per append from $O(\log n)$ to $O(1)$ for block authors. 
%
However, the main advantage of using $\MMB$ is its incrementality: when Alice resyncs after a $k$-increment, she only needs to sample $O(\log k)$ headers out of the $k$ new ones, and thus only handles proofs of size $O(\log k)$ with $\MMB$, which leads to a CCC of $O(\log^2 k)$, as opposed to a CCC of $O(\log n \cdot \log k)$ with $\MMR$.

\vspace{3mm}

\textbf{Proofs of Proof-of-Stake (PoPoS).} 
Consider a PoS blockchain network whose consensus protocol is run by a committee of validators that changes every epoch. 
PoPoS~\cite{agrawal2023proofs} proposes using an $\MMR$ to commit a list $X$ where each item $x_i$ in it contains information about an epoch's committee of validators, as well as their signed handover to the next epoch's committee. 
The authors argue that when Alice synchronizes to the network, her most security-critical step is retrieving the current commitment $\langle X_n \rangle$ from an honest source. 

Now assume Alice retrieves two different commitments $\langle X_n \rangle, \langle X'_n \rangle$, from an honest and an adversarial server, respectively. 
Authors propose an interactive \emph{bisection game}, for Alice to find the first index $i$ at which lists $X_n$ and $X_n'$ disagree; once she downloads items $x_i$ and $x'_i$, she can establish which list made a non-authorized committee handover. 
Starting from the root of the MMR (i.e., the commitment), in each communication round of the game, Alice requests from both servers both children of the current node, zooms in on the first child in which they disagree, and repeats. 
Eventually, she will reach a leaf, namely $h_i=H(x_i)$.  
Hence, the number of rounds equals the depth of the leaf, which is $O(\log n)$ with MMR. 

Again, we can exploit MMB's incrementality in this LCP. 
If Alice already knows a validator committee $x_m$ from a $k$-increment ago, $k=n-m$, her CCC can be reduced to $O(\log k)$ if using MMB. 
First, she requests and verifies the two membership proofs $\pi_{x_m\in X_n}$ and $\pi_{x_m\in X'_n}$. 
Using the terminology introduced in Section~\ref{s:increment}, we can say that each of these proofs contains a node set covering $[h_{m+1}, h_{n}]$, each of size $O(\log k)$.
Alice can zoom in on the first node for which these sets disagree, and start the bisection game from there. 
As each leaf in $[h_{m+1}, h_{n}]$ has a depth of $O(\log k)$ in MMB, the game will only take $O(\log k)$ rounds. 

\vspace{3mm}

\textbf{Cross-chain bridges.} 
A trustless communication bridge between a source blockchain network $A$ and a target network $B$ typically works as follows. 
Users in $A$ continually generate messages for $B$, such as transfer requests. 
Network $A$ commits to a list $X$ of these outgoing messages \emph{on the granularity of blocks}: 
namely, it first commits to messages generated within a block (e.g., with a Merkle tree), 
and then it appends this commitment to list $X$. 
Hence, we have two nested commitment schemes, and authenticating an outgoing message requires both a \emph{message inclusion proof} relative to a block, and a \emph{block inclusion proof} in $X$. 
In what follows, we only consider the changes to the block inclusion proof, when list $X$ is committed with MMB instead of MMR. 

A smart contract in $B$ receives from $A$ an update on commitment $\langle X \rangle$ at regular intervals, typically on the order of minutes or hours. 
Then, it receives messages (either from users or from a relayer), authenticates them against its copy of $\langle X \rangle$, and executes them. 
Hence, when a message is created in $A$, it can only be relayed the next time $\langle X \rangle$ is updated in $B$. 
Now, while this smart contract should be able to accept arbitrarily old messages, most users (or relayers) have no reason to wait longer than necessary, so the recency of most relayed messages will be of the same order of magnitude as the time interval between commitment updates, i.e., minutes or hours old. 
And since submitting a message with a large proof incurs high transaction fees, which can be significant in popular target networks such as Ethereum, it is sensible to let recent messages in $X$ have shorter proofs. 

We consider the case of Polkadot~\cite{burdges2020overview}. 
Since 2024, its validators run the \emph{Bridge efficiency enabling finality yielder} (BEEFY) protocol~\cite{bhatt2025trustless} on top of the regular consensus protocol. 
At regular intervals, they update and sign a commitment $\langle X \rangle$ to a list $X$ that has one item per finalized block. 
In turn, each item commits not only to a block, but also to any concurrent messages to external networks; thus, commitment $\langle X \rangle$ works as a universal digest of the state of Polkadot for all light clients and bridges. 
In particular, this commitment is used by Snowbridge~\cite{Snowbridge} and Hyperbridge~\cite{Hyperbridge}, trustless Polkadot-Ethereum bridges in production since 2024 and 2025, respectively. 

We remark that, as of late 2025, Polkadot has north of $28$ million blocks, so committing to them with a balanced tree would result in block inclusion proofs around 25 hashes long (and growing). 
In contrast, if BEEFY uses MMB, by Theorem~\ref{thm:Zipf} and the fact that Polkadot produces 10 blocks per minute, in expectation, the block inclusion proof of a twelve-hour old message ($k=7200$) would be less than $20$ hashes long, that of a one-hour old message ($k=600$) would be about $15$ hashes long, and that of a five-minute old message ($k=50$) would be about 10 hashes long.

\vspace{2mm}

\textbf{Stateless architecture, and MiniChain.}  
Many blockchain applications suffer from state bloat: they contain a large data set $X$ that keeps growing at a fast pace, while any individual item is used very rarely, or might even stop serving any purpose whatsoever soon after its creation. 
This bloat constitutes a burden for full nodes in terms of disk space, and slows down the synchronization process for new full nodes that join the network. 
A possible solution to this problem is to use a stateless architecture. 
The idea is to require full nodes to maintain a commitment scheme over list $X$, but relieve them of the responsibility of storing it in full. 
Instead, each user has the responsibility to store any item of interest, along with its membership proof, and update this proof whenever necessary. 

MiniChain~\cite{chen2020minichain} proposes a stateless architecture for a UTXO-based blockchain. 
Authors suggest using $\MMR$ to commit to the list $X$ of all transaction outputs (TXOs) \emph{on the granularity of blocks}. 
As before, this means there are two nested commitment schemes --~the internal one being a Merkle tree per block~-- and authenticating a TXO requires both a \emph{TXO inclusion proof} relative to a block, and a \emph{block inclusion proof} in $X$. 
Again, we only consider the changes to the block inclusion proof, when list $X$ is committed with an $\MMB$. 

There is a second append-only list, of \emph{spent} transaction outputs (STXOs), committed with an RSA-based scheme~\cite{li2007universal}.
Hence, in order to spend a TXO, a user submits the two aforementioned inclusion proofs to the TXO list $X$, and a proof of non-membership to the STXO list. 
Authors propose that, while full nodes run $\MMR$ on $X$ (which requires $O(\log n)$ memory), users only store block inclusion proofs up to the local mountain peak, i.e., $\UMMR$ proofs~\cite{reyzin2016efficient}. 
This way, users profit from low update frequency. 

\emph{$\MMB$ Variant 1.} We can require full nodes to run $\MMB$ --~store all $O(\log n)$ peaks and bagging nodes, and maintain an $O(1)$-sized commitment~-- while users only store $\UMMB$ proofs. 
This way, users' proofs are shorter and retain all asynchrony properties. 
Full nodes can validate any user's $\UMMB$ proof, and may ``extend'' it into an $\MMB$ proof. 

If list $X$ is instead committed with $\MMB$ Variant 1, each append update only takes constant time. 
Besides, users gain the property of recent-proof compatibility, so they need not be constantly online to check for proof updates; instead, they are given a guaranteed validity period, which grows exponentially fast after each update. 
And the CCC of each such update reduces from $O(\log n)$ to $O(1)$, as their proof can only grow by one or two hashes. 

Finally, we observe that a user typically spends a TXO soon after its creation. 
For example, it was recently proved that in Bitcoin, the lifetime of TXOs --~the number of blocks created in between a TXO being created and it being spent~-- follows a recency-based inverse power law distribution with exponent $s\approx 1.31$~\cite{hirata2024stochastic}. 
This means that over $\%90$ of Bitcoin TXOs are spent when they are less than $1000$ blocks old. And by Theorem~\ref{thm:Zipf} with $k=1000$, we see that $\MMB$ would give these TXOs a block inclusion proof less than $18$ hashes long in expectation. 
Hence, using a variant of $\MMB$ in MiniChain would likely provide most users with  transaction fees that are cheaper in expectation, and that remain constant over time, unaffected by the state bloat. 

\emph{MMB Variant 2.} We could require full nodes to also store internal mountain nodes of height (say) 20 or higher. 
Thus, users can eventually stop updating their proofs after (say) one year, once their leaves belong to a mountain of height 20. 
The memory requirements for full nodes will increase at a very slow linear rate of (say) two hashes per year. 

\vspace{3mm}

\textbf{Registration-based encryption.} 
In the context of public-key cryptography, consider a set of $n$ users, each of whom has an ID (phone number, email, domain name, etc.). 
Traditionally, each user generates a (secret key, public key) pair, and if Alice wants to send a message to Bob, she first needs to learn his public key, which requires the existence of public-key infrastructure (PKI), such as an authority who provides public keys under request. 

Identity-based encryption (IBE)~\cite{shamir1984identity} is an alternative where users do not generate a (secret key, public key) pair;  
instead, an authority manages a single master key. 
Alice only needs to know the master key and Bob's ID to send him an encrypted message, and Bob can request his personal decryption key from the authority. 
IBE reduces the communication needed between users and the authority, as each user only needs to request a single key (their own decryption key), but introduces the \emph{key escrow problem}: the authority may decrypt any of Bob's messages, or give his decryption key to third parties. 

Garg et al.~\cite{garg2018registration} recently proposed a third alternative called registration-based encryption (RBE): each user generates a (secret key, public key) pair, and the authority simply commits to the list $X$ of users' (ID, public key) pairs. 
Alice only needs to know the public commitment $\langle X\rangle $ and Bob's ID to send him an encrypted message. 
Bob, in turn, needs his secret key and the membership proof of his item in $X$ to decrypt the message. 

The use of $\UMMR$ is proposed in~\cite{garg2018registration} to achieve low update frequency, so that Bob only needs to communicate with the authority $O(\log n)$ times --~to update his membership proof~-- after $n$ new user registrations. 
If $\UMMB$ is used instead, Bob also gains the property of recent-proof compatibility. 
But more importantly, \emph{the CCC of any interaction with the authority reduces from $O(\log n)$ to $O(1)$}: namely, any user registration (i.e., item append) and any decryption key update (as a membership proof can only grow by one or two hashes).

The use of RBE in blockchain has been recently discussed~\cite{Glaeser}, where the role of the authority is filled by a smart contract, and each new user bears the computational cost needed to update the commitment for their registration. 
As computations can be prohibitively expensive on smart contracts, it becomes vital to have a registration with $O(1)$ cost.

\vspace{3mm}

\textbf{Low update frequency.} Motivated by the use of Merkle structures in registration-based encryption, recent work \cite{mahmoody2022lower, mahmoody2023online, bonneau2025merkle, qi2025tight} has been done on the theoretical limits of commitment schemes in terms of proof update frequencies. 
In particular, Bonneau et al.~\cite{bonneau2025merkle} considered the number of times that the membership proof of an item needs to be updated, \emph{on average}, if we start with an empty list and append $n$ items in sequence. 
They remarked that $\UMMR$ achieves an average of $O(\log n)$, which is smallest among all known binary Merkle structures with polylogarithmic sized commitments, and proved a theoretical lower bound of $\omega(1)$. 
Qi~\cite{qi2025tight} recently improved this lower bound to $O(\log n / \log \log n)$, for any append-only commitment scheme with a polylogarithmic sized commitment. 
This latter bound matches the average reached by a non-binary Merkle structure presented in~\cite{mahmoody2023online}. 

We highlight that $\UMMB$ achieves an average of $O(\log n)$ updates, matching $\UMMR$ in the metric above, while reducing the complexity of each such update from $O(\log n)$ to $O(1)$.  
Besides, $\UMMB$ is the first scheme to achieve the property of recent-proof compatibility. 
We argue that, in many applications, this property is both more useful for protocol designers and easier to understand for users, than simply having a bound on the average update frequency across all items and all states of the list. 
Namely, a recent-proof compatibility with parameter $c=1/5$ means that if Alice updates the membership proof of an item five hours after its creation, this proof should remain valid for about one hour, regardless of the list size or the append rate, as long as this rate is roughly constant or changes slowly.
For example, such a guarantee would be very useful to users of the MiniChain application described above. 


\emph{MMB Variant 3.} In U-MMB (or MMB Variant 1), if we require full nodes to also store the left-hand child of every peak, the recent-proof-compatibility parameter enjoyed by users increases from $c=1/5$ to $c=1/3$. 
This can be shown by adapting the first case in the proof of Lemma~\ref{lem:rpc} (now, the proof becomes invalid only after the corresponding leaf participates in two merges, instead of one). 
Parameter $c$ can be further improved by having verifiers store more descendants of the peaks.

\section{Algorithmic details}\label{s:alg-UMMB}

In this section we provide algorithmic details for the $\UMMB$ commitment scheme. 
In particular, we demonstrate that the manager's append operation can be performed in constant time and logarithmic cache memory, without any large hidden factors that may render the scheme impractical. For further algorithmic details on $\UMMB$ and $\MMB$, we invite the reader to visit our proof-of-concept implementation in Clojure~\cite{Hambrock}. 

The $\UMMB$ scheme is built on top of a background application with an append-only ordered item list $X=(x_1, x_2, \cdots, x_n)$.  
It is composed of the following data structures:
\begin{itemize}
    \item A counter $n$ of the current number of items in $X$.
    \item \texttt{Peaks}: A list of mountain peaks, formally defined as a singly linked list of $t=\lfloor \log_2(n+1)\rfloor$ ``peak'' objects, where in turn each peak object \texttt{P} has the following fields: 
    \begin{itemize}
        \item \texttt{P.hash}: the hash of the peak node,
        \item \texttt{P.height}: the height of the corresponding mountain, and
        \item \texttt{P.prev}: a pointer to the previous peak object to the left.
    \end{itemize}
    We also keep a pointer $\texttt{P}_\texttt{head}$ to the rightmost peak object at all times.
    \item \texttt{Pairs}: a list representing all mergeable pairs of mountains, formally defined as a stack 
    of pointers to peak objects, containing the right-hand peak of each mergeable pair, and kept sorted so that the rightmost peak is on top.  
    \item \texttt{Hashes[]}: the hashes of nodes in mountains, including leaves and non-leaves, formally defined as an append-only array of size $2n+2=O(n)$ whose entries are hash values.  
\end{itemize}

In $\UMMB$, commitment $\langle X\rangle$ corresponds to the \texttt{Peaks} list, which can be accessed via the $\texttt{P}_\texttt{head}$ pointer and the $\texttt{P.prev}$ pointers used in sequence. 
All of these structures appear on a public bulletin board kept up-to-date by a scheme manager. 
Participants --~i.e., light clients~-- keep in memory the first two data structures, and the manager keeps in memory the first three data structures, hence they all need $O(\log n)$ memory. 

Array \texttt{Hashes[]} can be saved in permanent storage. 
It is append-only: entries are written in strict sequence and never erased nor modified. Its size grows dynamically as needed.  
The manager only needs append rights to it --~not read rights~-- and we assume that each append takes constant time. 
In turn, a participant only needs read rights, and we assume that reading entry \texttt{Hashes[$i$]} for a given index $i$ also takes constant time.

We present the manager's append operation in Algorithm~\ref{alg:append}. 
For an empty set $X$, the input is initialized as $n=0$, $\texttt{P}_\texttt{head}=\texttt{Null}$, $\texttt{Pairs} = \emptyset$, $\texttt{Hashes[]} = [0,0,0]$.

\begin{algorithm}[htp]
    \caption{Append operation, performed by the manager. The new item has hash $h$.}
    \label{alg:append}
 
\begin{algorithmic}
\State \textbf{Input:} Counter $n$, list \texttt{Peaks}, peak pointer $\texttt{P}_\texttt{head}$, stack \texttt{Pairs}, array \texttt{Hashes[]}, hash $h$
\State Update counter $n++$
\State Create new peak object $\texttt{P}_\texttt{new}$ \Comment{Start of ADD step}
\State Set $\texttt{P}_\texttt{new}.\texttt{hash}\leftarrow h$, $\texttt{P}_\texttt{new}.\texttt{height}\leftarrow 0$, $\texttt{P}_\texttt{new}.\texttt{prev}\leftarrow \texttt{P}_\texttt{head}$
\If{($\texttt{P}_\texttt{head}\texttt{ $\neq$ Null}$ \&  $\texttt{P}_\texttt{head}.\texttt{height}==0$)} 
    \State $\texttt{Pairs.push(P}_\texttt{new}\texttt{)}$ \Comment{new mergeable pair}
\EndIf
\State $\texttt{Peaks.add(P}_\texttt{new}\texttt{)}$, set $\texttt{P}_\texttt{head} \leftarrow \texttt{P}_\texttt{new}$
\State Set $\texttt{Hashes}[2n+1]\leftarrow h$ \Comment{append hash to array}
\If{($|\texttt{Pairs}|>0$)} \Comment{Start of MERGE step}
    \State Let $\texttt{P}_\texttt{mrg}\leftarrow \texttt{Pairs.pop()}$ 
    \State Let $\texttt{P}_\texttt{garbage} = \texttt{Peaks.remove}(\texttt{P}_\texttt{mrg}.\texttt{prev})$ \Comment {remove left-hand peak%
    \footnotemark} 
    \State Update $\texttt{P}_\texttt{mrg}.\texttt{hash}\leftarrow 
    H(\texttt{P}_\texttt{garbage}.\texttt{hash} || \texttt{P}_\texttt{mrg}.\texttt{hash})$ \Comment{$\texttt{P}_\texttt{mrg}$ becomes new parent}
    \State Update $\texttt{P}_\texttt{mrg}.\texttt{height}++$, $\texttt{P}_\texttt{mrg}.\texttt{prev}\leftarrow  \texttt{P}_\texttt{garbage}.\texttt{prev}$
    \If{ ($\texttt{P}_\texttt{mrg}.\texttt{prev $\neq$ Null}$ \& $\texttt{P}_\texttt{mrg}.\texttt{prev.height $==$ P}_\texttt{mrg}.\texttt{height}$) } 
        \State $\texttt{Pairs.push(P}_\texttt{mrg}\texttt{)}$ \Comment{new mergeable pair}
    \EndIf
    \State Set $\texttt{Hashes}[2n+2]\leftarrow \texttt{P}_\texttt{mrg}.\texttt{hash}$ \Comment{append hash to array}
\EndIf
\Return
\end{algorithmic} 
\end{algorithm}
\footnotetext{The $\texttt{Peaks.remove()}$ operation is assumed to be a garbage collection process that removes a peak object from the $\texttt{Peaks}$ list, but makes its contents temporarily available in the variable $\texttt{P}_\texttt{garbage}$}

Evidently, the algorithm requires $O(\log n)$ cache memory and $O(1)$ time, which proves Lemma~\ref{lem:U-append}. 
We remark that hashes are appended to array \texttt{Hashes[]} in the order they are created, namely a leaf followed by a merge peak, and in case there is no merge step we skip an entry (or alternatively append a default value~$\bullet$), so that the hash of the $n$-th item $x_n$ is always stored in \texttt{Hashes[$2n+1$]}. 
By doing so, we facilitate index calculations for a node's parent and children, as we show next.
Recall from Section~\ref{s:unbagged} that the height of a mountain node is defined as its distance to any descendant leaf, and that for a positive integer $m$, its 2-adic valuation $\nu_2(m)$ is the highest integer exponent $\nu$ such that $2^{\nu}$ divides $m$, $i\%m$ is the remainder of $i$ modulo $m$, and $i\backslash m = \lfloor i/m \rfloor$ is the integer division of $i$ over $m$.

\begin{lemma}\label{lem:array}
Array \texttt{Hashes[]} has the following properties:
\begin{enumerate}
    \item Entry \texttt{Hashes[$i$]} is skipped if and only if $i$ is zero or a power of two, so $O(\log n)$ entries are skipped.
    \item If $N$ is the node whose hash is stored at \texttt{Hashes[i]}, then $N$ is a leaf if and only if $i$ is odd, and in that case it corresponds to the $(i\backslash 2)$-th item. 
    More generally, the height of $N$ never changes and is exactly $\nu_2(i)$.
    \item If $\texttt{isLeft}(i):=(i\backslash 2^{\nu_2(i)+1})\% 2$, then $N$ is (or will be) a left child if $\texttt{isLeft}(i)=1$ and a right child if $\texttt{isLeft}(i)=0$, and its sibling is (or will be) stored at index $\texttt{sibling}(i):=i+(-1)^{\texttt{isLeft}(i)+1}\cdot 2^{\nu_2(i)+1}$.
    \item The parent of node $N$ is (or will be) stored at index $\texttt{parent}(i):=i+(1+2\cdot \texttt{isLeft}(i))\cdot 2^{\nu_2(i)}$, and $N$ is a peak if and only if $\texttt{parent}(i)>2n+2$.
    \item If node $N$ is not a leaf, then its left and right children are respectively stored at indices $\texttt{left}(i):=i-3\cdot 2^{\nu_2(i)-1}$ and $\texttt{right}(i):=i-2^{\nu_2(i)-1}$. 
\end{enumerate}
\end{lemma}

\begin{proof}
The first claim follows from Lemma~\ref{lem:Sn}. 
Next, in our scheme mountains evolve only via the add and merge steps of the append operation; as neither step changes the height of an existing node, it follows that node heights never change. 
When the $n$-th node is appended, its leaf is stored at position $i=2n+1$, which is odd, hence $\nu_2(i)=0$ matching the node height. 
When there is a merge step, two mountains of height $j:=\nu_2(n+1)$ are merged (by Lemma~\ref{lem:Sn}) via a merge peak of height $j+1$ that is stored at position $2n+2$, and hence $\nu_2(2n+2)=\nu_2(2\cdot(n+1))=1+\nu_2(n+1)=j+1$, again matching the node height. This proves the second claim.

As a consequence of the second claim, we have that for any $s\geq 0$, hashes of nodes of height $s$ are stored in ascending order at indices $3\cdot 2^s, 5\cdot 2^s, 7\cdot 2^s, 9\cdot 2^s, \cdots$, because we use all such indices that are not a power of two. 
As it follows from Lemma~\ref{lem:Sn} that there can be at most two peaks of height $s$ at any given moment, it must be the case that peaks at indices $3\cdot 2^s$ and $5\cdot 2^s$ are merged together, then peaks at indices $7\cdot 2^s$ and $9\cdot 2^s$ are merged together, and so on. This proves the third claim. 
Finally, the nodes of height $s+1$ resulting from these merges are also created in ascending order, hence the parent of nodes at indices $3\cdot 2^s$ and $5\cdot 2^s$ must be at index $3\cdot 2^{s+1}$, the parent of nodes at indices $7\cdot 2^s$ and $9\cdot 2^s$ must be at index $5\cdot 2^{s+1}$, and so on. This proves the last two claims.
\end{proof}

\begin{figure}[htp]
\centering
\includegraphics[width=0.7\textwidth]{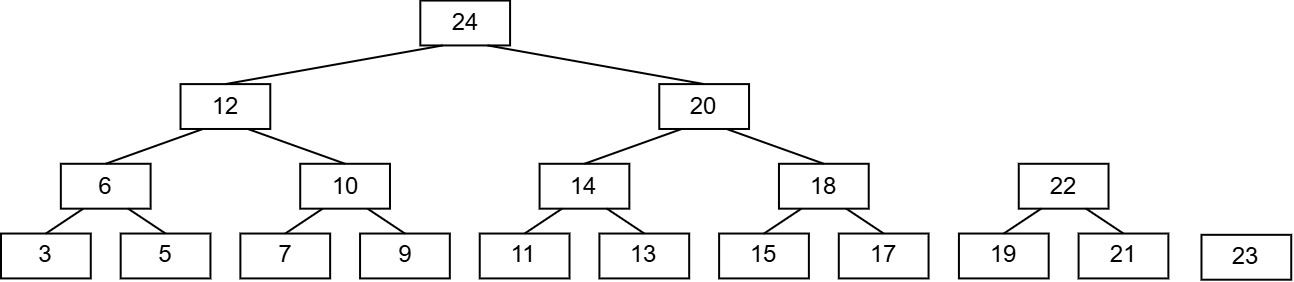}
\caption{U-MMB structure for $n=11$, with height sequence $S_{11}=(3,1,0)$. Each node is labeled with the index $i$ of array \texttt{Hashes[]} where its hash is stored.}
\label{fig:indices}
\end{figure}

For example, if $N$ is the node whose hash is stored at \texttt{Hashes[18]}, then its height is $\nu_2(18)=1$, and it is a right child because $\texttt{isLeft}(18)=(18\backslash 2^2)\% 2 =0$. 
Its sibling is at index $\texttt{sibling}(18)=18-2^2=14$, its parent at index $\texttt{parent}(18)=18+2^1 =20$, and its children at indices $\texttt{left}(18)=18-3\cdot 2^0=15$ and $\texttt{right}(18)=18-2^0=17$. 
See Figure~\ref{fig:indices}. 
Another way to formulate statements 2 and 3 in Lemma~\ref{lem:array} is that if $i$ in binary is $(...b_2b_1b_0)_2$, and $s$ is the lowest index such that $b_s=1$, then node $N$ is of height $s$, and is a left child if and only if $b_{s+1}=1$.

Finally, we describe how Alice can exploit these data structures to perform her update operations efficiently after being offline between states $X_m$ and $X_n$, where $k=n-m$. 
To update her copy of the commitment, she only needs to retrieve some $O(\log k)$ rightmost peaks from $\langle X_n \rangle$, which she can do via the $\texttt{P}_\texttt{head}$ pointer and each peak's pointer to the previous peak. 
Next, for any item $x_i\in X$ of interest, Alice can efficiently find its leaf hash at index $j=2i+1$ of the \texttt{Hashes[]} array, and can retrieve the indices of the hashes that form its membership proof $\pi_{x_i\in X_n}$ by recursively applying the formula $j'=\texttt{sibling}(\texttt{parent}(j))$, until this formula exceeds $2n+1$. 
Hence, assuming that when Alice goes offline at state $X_m$ she stores the first such index $j$ that exceeds $2m+1$, she knows the exact value of $n$ at which her stored proof $\pi_{x_i\in X_m}$ becomes outdated, and the exact indices of the $O(\log k)$ missing hashes to update the proof. 

\section{Amortized membership proof sizes} \label{s:amortized}

For a Merkle structure $\mathcal{M}$ over an append-only list $X$, and any integers $1\leq k\leq n$, let $\sigma_{\mathcal{M}}(k,n)$ be the size, measured in hashes, of the membership proof of the $k$-th most recently appended item when $|X|=n$, i.e., of $\pi_{x_{n-k+1}\in X_n}$.  
In this section we study the \emph{amortized} size of this proof, when $k$ is kept fixed and $n$ sweeps over a large interval. 
Formally, 
\begin{equation}\label{eq:amortized}
\bar{\sigma}_{\mathcal{M}}(k):= \lim_{N\rightarrow \infty} \frac{1}{N} \sum_{n=k}^{k+N-1} \sigma_{\mathcal{M}}(k,n).
\end{equation}

This criterion is important for a light client that regularly queries recent items, such as a user of a cross-chain bridge (see Appendix~\ref{s:apps}), or a user who wants to verify if their latest transaction was finalized. 
It is easy to see that while $\bar{\sigma}_{\chain}(k)=k$, $\bar{\sigma}_{\MMR}(k)$ diverges; hence, according to this criterion, even the Merkle chain eventually outperforms $\MMR$. 
Interestingly, the forward-bagged variant of $\MMR$ achieves a reasonable amortized size of $\bar{\sigma}_{\FMMR}(k)=\Theta(\log k)$, although its worst-case size diverges. 

In this section we prove the first statement of Theorem~\ref{thm:Zipf}, and prove as well that, on top of their superiority in terms of worst-case bounds, our structures are also superior in terms of amortized proof sizes. 
In particular, $\UMMB$ outperforms $\UMMR$, and $\MMB$ outperforms $\FMMR$; see Figure~\ref{fig:amortized}. 

\begin{figure}[htp]
    \centering
    \includegraphics[width=0.9\linewidth]{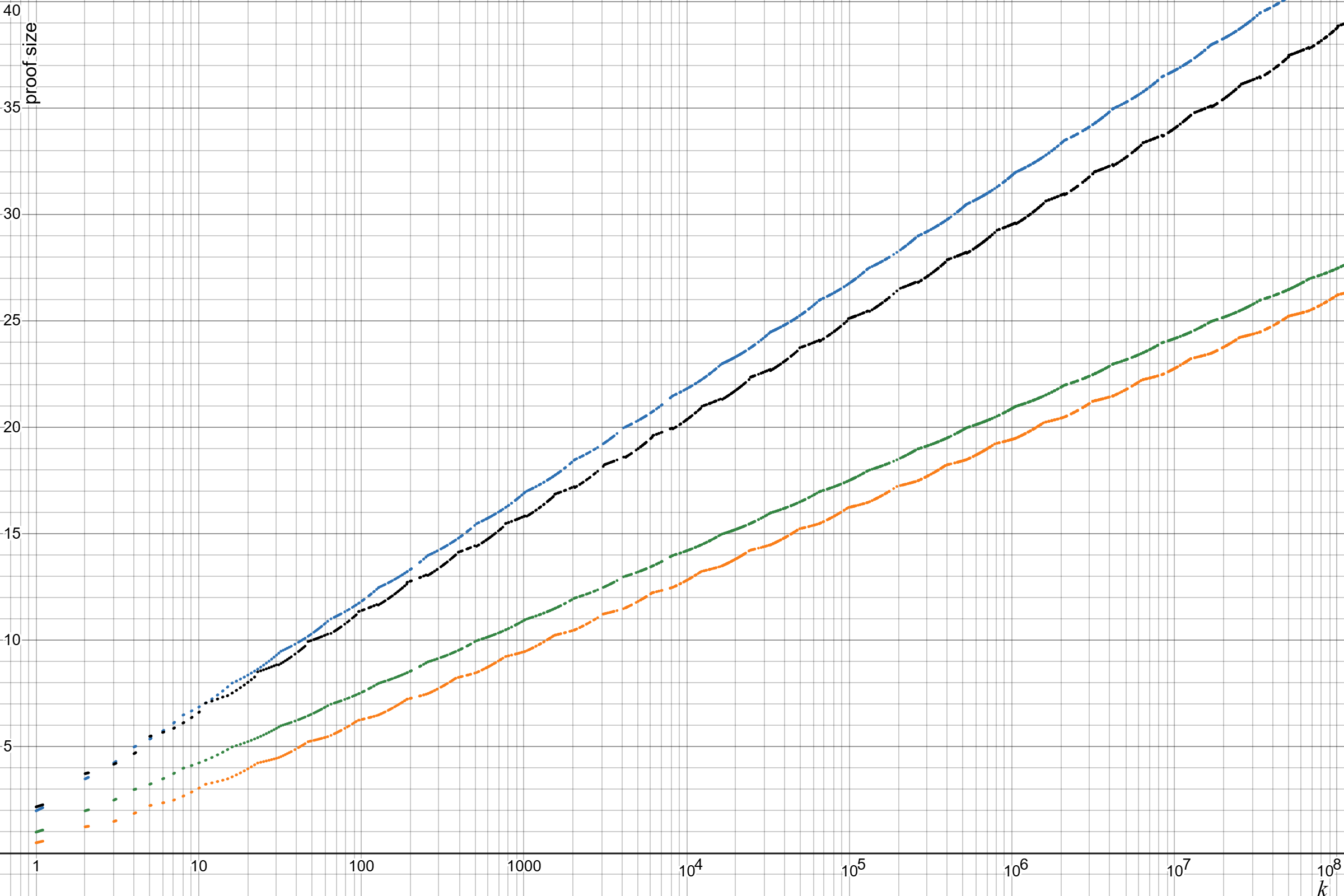}
    \caption{Amortized proof size $\bar{\sigma}_{\mathcal{M}}(k)$ of the $k$-th newest item: exact functions for $\UMMB$ (in orange), $\UMMR$ (in green) and $\FMMR$ (in blue), and an upper bound for $\MMB$ (in black). }
    \label{fig:amortized}
\end{figure}

\subsection{U-MMR and F-MMR}\label{s:am-MMR}

For $1\leq k\leq n$, consider the $\FMMR$ structure when $|X|=n$, let $r_{\FMMR}(k,n)$ be the number of ancestors of the $k$-th newest leaf that are range nodes, and let the indicator function $r'_{\FMMR}(k,n)$ be $1$ if this leaf is contained in the leftmost tree, and $0$ otherwise.  
Then,  
$$\sigma_{\FMMR}(k,n)=\sigma_{\UMMR}(k,n)+r_{\FMMR}(k,n)-r'_{\FMMR}(k,n),$$ 
because, as we extend a membership proof from $\UMMR$ to $\FMMR$, we add one hash per range node ancestor, except if the leaf is in the leftmost tree ($r'_{\FMMR}(k,n)=1$), in which case we can save one hash in the proof because the mountain peak has no sibling node. 
For the auxiliary functions defined above, we also define their amortized values with formulas similar to~\eqref{eq:amortized}. 
It is easy to see that $\bar{r'}_{\FMMR}(k)=0$, so 
\begin{equation}\label{eq:am-MMR}
\bar{\sigma}_{\FMMR}(k)=\bar{\sigma}_{\UMMR}(k)+\bar{r}_{\FMMR}(k).     
\end{equation}

We begin with a needed observation, whose proof is given in Appendix~\ref{s:proofs}. 

\begin{lemma}\label{lem:nu} 
The amortized value of the $2$-adic valuation function $\nu_2(n)$ is $$\lim_{N\rightarrow \infty} \frac{1}{N} \sum_{n=1}^N \nu_2(n)=1.$$
\end{lemma}

\begin{lemma}\label{lem:aUMMR}
In $\UMMR$, for any $k\geq 1$, if $d'=\lceil \log_2 k \rceil$ the amortized proof size is 
$$\bar{\sigma}_{\UMMR}(k)=d'+\frac{2k}{2^{d'}}-1=\log_2 k + O(1).$$
\end{lemma}
\begin{proof}
Fix $k\geq 1$ and let $d'=\lceil \log_2 k \rceil$. 
We decompose $n$ as $n=n_1 \cdot 2^{d'}+n_2$, consider $n_1$ and $n_2$ as independent random integers sampled uniformly from $[1,N]$ and $[0, 2^{d'}-1]$ respectively, and then we let $N\rightarrow \infty$. 
If $M$ is the mountain in $\UMMR$ that contains the $k$-th newest leaf, its height must be at least $d'-1$, because mountains of height at most $d'-2$ contain at most $\sum_{i=0}^{d'-2} 2^i = 2^{d'-1}-1<k$ leaves. 
In fact, the number of leaves in mountains of height at most $d'-1$ is exactly $n_2$, so we need to consider whether or not $n_2\geq k$.  

Case 1: $k\leq n_2< 2^{d'}$ (probability $1-\frac{k}{2^{d'}}$). 
In this case, the height of $M$ is $d'-1$. 

Case 2: $0\leq n_2 < k$ (pr. $\frac{k}{2^{d'}}$). 
$M$~is the smallest existing mountain of height at least $d'$. 
Its height is exactly $\nu_2(n_1\cdot 2^{d'})=d'+\nu_2(n_1)$, and by Lemma~\ref{lem:nu}, its amortized value is $d'+1$. 

The amortized proof size $\bar{\sigma}_{\UMMR}(k)$ is the expected height of $M$, so 
$$\bar{\sigma}_{\UMMR}(k)=\left(1-\frac{k}{2^{d'}}\right)(d'-1)+\frac{k}{2^{d'}}(d'+1) 
= d'+\frac{2k}{2^{d'}}-1,$$ 
as claimed.
\end{proof}

\begin{corollary}\label{cor:aFMMR}
In $\FMMR$, for any $k\geq 1$, if $d'=\lceil \log_2 k \rceil$ then 
\begin{align*}
\bar{r}_{\FMMR}(k) &= \frac{1}{2}d'+\frac{k}{2^{d'}} & \text{and} \\
\bar{\sigma}_{\FMMR}(k) &= \frac{3}{2}d'+\frac{3k}{2^{d'}}-1 = \frac{3}{2} \log_2 k + O(1). 
\end{align*}
\end{corollary}

\begin{proof}
In the previous proof, if $t$ is the number of ``ones'' in the binary representation of $n_2$, then mountain $M$ has $t$ range node ancestors in Case~1, and $t+1$ in Case~2. 
Hence, 
$$\bar{r}_{\FMMR}(k) = \E_{n_2}[t] + 1\cdot \Pr[\text{Case 2}] = \frac{1}{2}d'+\frac{k}{2^{d'}},$$
where $\E_{n_2}[t]=d'/2$ because the binary representation of $n_2$ can be modeled as $d'$ independent fair coin tosses. 
The second claim follows from identity~\eqref{eq:am-MMR}. 
\end{proof}

\subsection{U-MMB and F-MMB}

For $1\leq k\leq n$, consider the $\FMMB$ structure when $|X|=n$, let $r_F (k,n)$ be the number of ancestors of the $k$-th newest leaf that are range nodes, and let $r'_F (k,n)$ be $1$ if this leaf is contained in the leftmost tree, and $0$ otherwise. 
As explained in Section~\ref{s:am-MMR}, it follows that 
$$\sigma_{\FMMB}(k,n) = \sigma_{\UMMB}(k,n) + r_F(k,n)-r'_F(k,n).$$
We define the amortized values of these functions with formulas similar to~\eqref{eq:amortized}. It is easy to see that $\bar{r'}_F(k)=0$, so 
\begin{equation}\label{eq:am-FMMB}
\bar{\sigma}_{\FMMB}(k)=\bar{\sigma}_{\UMMB}(k)+\bar{r}_F(k).
\end{equation}

We will show that the function $r_F(k,n)$ above is \emph{periodic}. 
We say that a function $f(k,n)$ is periodic, with period $N=N(k)$, if $f(k,n)=f(k,n+N)$ for any $1\leq k\leq n$.%
\footnote{We do not include any sense of minimality in our definition of a period. In the rest of the section we establish periods for many auxiliary functions, without any claims that the given periods are minimal.}
Clearly, its amortized value matches its average value within any interval of period size, i.e., 
\begin{equation}\label{eq:period}
  \frac{1}{N} \sum_{n=m}^{m+N-1} f(n,k) = \bar{f}(k), \quad \text{for any } m\geq k.  
\end{equation}

We will use the notation and results from Lemma~\ref{lem:Sn}: 
we let $(\cdots b_2b_1b_0)$ be the binary representation of $n+1$ where $n=|X|$, we enumerate the mountain positions in $\UMMB$ starting from zero from the right, and we recall that the mountain at position $i$ has a height $s_i = i +b_i$.  
We begin with a needed technical observation. 

\begin{lemma}\label{lem:MM'}
For any $1\leq k< n$, if U-MMB has $n$ leaves, $d=\lfloor \log_2 (k+1) \rfloor$, and $M$ and $M'$ are the mountains that contain the $k$-th and $(k+1)$-th newest leaves, respectively, then 
\begin{itemize}
    \item $M$ is at position $d-1$ if $k+1-2^d \leq (n+1 \mod 2^d)$, and position $d$ otherwise; 
    \item $M\neq M'$ if and only if $n \equiv k \mod 2^d $; and furthermore 
    \item there is a range split between $M$ and $M'$ if and only if 
    $$\begin{cases}
        n \equiv k \mod 2^{d+1} & \text{for } 2^d\leq k+1< \frac{3}{2} 2^d \\
        n \equiv k \mod 2^{d+2} & \text{for } \frac{3}{2} 2^d \leq k+1 < 2^{d+1}.
    \end{cases}$$
\end{itemize}
\end{lemma}

\begin{proof}
Mountains up to position $d$ have at least $\sum_{i=0}^d 2^{s_i} \geq \sum_{i=0}^d 2^{i}=2^{d+1}-1\geq k+1$ leaves, while mountains up to position $d-2$ have at most $\sum_{i=0}^{d-2} 2^{s_i} \leq \sum_{i=0}^{d-2} 2^{i+1}=2^{d}-2< k$ leaves. 
This proves that both $M$ and $M'$ must be at positions either $d-1$ or $d$. 
In fact, the number of leaves in mountains up to position $d-1$ is 
$$\sum_{i=0}^{d-1} 2^{s_i} = \sum_{i=0}^{d-1} 2^{i+b_i} = (2^d-1)+\sum_{i=0}^{d-1} 2^{b_i}=(2^d-1)+n', \quad \text{where } n':=(n+1\mod 2^{d}).$$ 
Hence, $M$ is at position $d-1$ if and only if this number of leaves is $2^d-1+n'\geq k$, and mountains $M$ and $M'$ are different (at positions $d-1$ and $d$, respectively) if and only if this inequality is tight, which implies $n \equiv k \mod 2^d$. 
This proves the first two claims. 

To prove the last claim, let $(\cdots b'_2b'_1b'_0)$ be the binary representation of $k+1$ and assume that $n+1 \equiv k+1 \mod 2^d$, so $M$ and $M'$ are at positions $d-1$ and $d$, respectively, by the second claim. 
Recall from the range splitting conditions in Section~\ref{s:MMB-append} that there is a split between $M$ and $M'$ if and only if $(b_{d+1} b_d b_{d-1})\in\{01|0, 01|1, 11|0\}$, i.e., these three bits of $n+1$ fall into one of these three configurations. 
We consider two cases: 

Case 1: $2^d\leq k+1< \frac{3}{2}2^d$. 
We have that $(b'_d b'_{d-1})=10$, so $b_{d-1}=b'_{d-1}=0$ by our assumption. 
Conditioned to this assumption, there is a range split if and only if $(b_db_{d-1})=1|0$, implying an additional match on bits $b_d=b'_d=1$, so $n+1\equiv k+1 \mod 2^{d+1}$.

Case 2: $\frac{3}{2}2^d \leq k+1 < 2^{d+1}$. 
Now we have $(b'_{d+1}b'_d b'_{d-1})=011$, so $b_{d-1}=b'_{d-1}=1$ by our assumption.  
Conditioned to it, there is a range split if and only if $(b_{d+1} b_d b_{d-1})=01|1$, implying additional matches on bits $b_{d+1}=b'_{d+1}=0$ and $b_d=b'_d=1$, so $n+1\equiv k+1 \mod 2^{d+2}$.
\end{proof}

\begin{lemma}\label{lem:aUMMB}
In U-MMB, for any $k\geq 1$, if $d=\lfloor \log_2 (k+1) \rfloor$ the amortized proof size is
\begin{align*}
\bar{\sigma}_{\UMMB}(k)&= \begin{cases}
    d + \frac{3(k+1)}{2^{d+1}}-2 & \text{for} \quad 2^d\leq k+1< \frac{3}{2} 2^d \\
    d+\frac{k+1}{2^{d+1}}-\frac{1}{2} & \text{for} \quad \frac{3}{2} 2^d\leq k+1< 2^{d+1} .  
\end{cases} 
\end{align*} 
\end{lemma}

\begin{proof}
Let $d=\lfloor \log_2 (k+1) \rfloor$. 
We consider bits $b_0$ through $b_{d}$ in the binary representation of $n+1$ as $d+1$ independent fair coin tosses. 
We will compute the expected height of the mountain $M$ that contains the $k$-th most recent leaf. 
By the previous lemma, $M$ is at position $d-1$ or $d$, depending on whether $k+1-2^d \leq n'$, where $n':=n+1 \mod 2^d$: 

Case A: $0\leq n'< k+1-2^{d}$ (probability $\frac{k+1}{2^{d}}-1$). 
Mountain $M$ is at position $d$. 
Its height, $s_{d}=d+b_{d}$, is either $d$ or $d+1$, depending on bit $b_{d}$. 
As this bit corresponds to a fair coin toss, the expected height of $M$ is $d+1/2$. 

Case B: $k+1-2^{d} \leq n'< 2^{d}$ (prob. $2-\frac{k+1}{2^{d}}$). 
$M$~is at position $d-1$. 
Its height, $s_{d-1}=d-1+b_{d-1}$, is either $d-1$ or $d$, depending on bit $b_{d-1}$, i.e., whether or not $n'< 2^{d-1}$:  
\begin{itemize}
    \item Case B.1: $k+1-2^{d}\leq n'< 2^{d-1}$ (prob. $\max\{0, \frac{3}{2}-\frac{k+1}{2^{d}}\}$). 
Its height is $s_{d-1}=d-1$. 
    \item Case B.2: $\max\{k+1-2^{d}, 2^{d-1}\}\leq n'< 2^{d}$ (prob. $\min\{\frac{1}{2}, 2-\frac{k+1}{2^{d}}\}$). 
$s_{d-1}=d$.
\end{itemize}

In order to aggregate these results, we consider two cases for the value of $k$.  
If $k+1< \frac{3}{2}2^d$, the expected height of $M$ (by Cases A, B.1 and B.2, respectively) is 
\begin{align*}
    \left(\frac{k+1}{2^{d}} - 1\right)\left(d+\frac{1}{2}\right)+\left( \frac{3}{2}-\frac{k+1}{2^{d}}\right)(d-1)+\frac{1}{2}\cdot d  
  = d + \frac{3(k+1)}{2^{d+1}} - 2,
\end{align*}
while if $k+1\geq \frac{3}{2}2^d$, Case B.1 has probability zero, and the expected height of $M$ is 
\begin{align*}
    \left(\frac{k+1}{2^{d}} - 1\right)\left(d+\frac{1}{2}\right)+\left( 2-\frac{k+1}{2^{d}} \right) \cdot d 
  = d + \frac{k+1}{2^{d+1}} -\frac{1}{2},
\end{align*}
as claimed. 
\end{proof}

\begin{corollary}
In $\FMMB$, for any $k\geq 1$, if $d=\lfloor \log_2 (k+1) \rfloor$ then 
\begin{align*}
\bar{r}_F(k) &= d+\frac{k+1}{2^{d}}-1, \quad \text{and} \\
\bar{\sigma}_{\FMMB}(k) &= \begin{cases}
    2d + \frac{5(k+1)}{2^{d+1}} - 3 & \text{for} \quad 2^d\leq k+1< \frac{3}{2}2^d \\
    2d+\frac{3(k+1)}{2^{d+1}}-\frac{3}{2} & \text{for} \quad \frac{3}{2} 2^d\leq k+1< 2^{d+1}, 
\end{cases} 
\end{align*}
\end{corollary}

\begin{proof} 
From the previous proof, we see that $\bar{r}_F(k)$ is one unit more than the expected position of $M$ in the mountain list (because we numbered positions from zero). 
Hence, it is $d+1$ in Case A and $d$ in Case B, with an average of 
\begin{align*}
\bar{r}_{F}(k) &= \left(\frac{k+1}{2^{d}} - 1\right)(d+1) + \left(2-\frac{k+1}{2^{d}} \right) \cdot d = d+ \frac{k+1}{2^d}-1, 
\end{align*} 
as claimed. The second claim follows from identity~\eqref{eq:am-FMMB}. 
\end{proof}

We remark that $\UMMB$ beats $\UMMR$ in terms of amortized proof sizes: 
see Figure~\ref{fig:amortized}, where we plot in orange and green the functions $\bar{\sigma}_{\UMMB}(k)$ and $\bar{\sigma}_{\UMMR}(k)$, from Lemmas \ref{lem:aUMMB} and \ref{lem:aUMMR}, respectively. 
In fact, the difference $\bar{\sigma}_{\UMMR}(k)-\bar{\sigma}_{\UMMB}(k)$ is at least $1/2$ for $k\geq 1$, and at least $1$ for $k\geq 3$. 
In Appendix~\ref{s:proofs} we provide a proof of the following observation. 

\begin{lemma}\label{lem:diffU}
For any $k\geq 1$, we have $\bar{\sigma}_{\UMMR}(k)-\bar{\sigma}_{\UMMB}(k) \geq \frac{5}{4}-\frac{3}{2(k+1)}$.
\end{lemma}

\subsection{MMB}

Let $h$ be the $k$-th newest leaf in the $\MMB$ structure for some $1\leq k\leq n=|X|$. 
Let $r(k,n)$ and $b(k,n)$ be the number of ancestors of $h$ that are range nodes and belt nodes, respectively, let $r'(k,n)$ be $1$ if $h$ is contained in the leftmost mountain of its range, $0$ otherwise, and let $b'(k,n)$ be $1$ if $h$ is contained in the leftmost range, $0$ otherwise. 
Then, 
\begin{equation}\label{eq:kn-MMB}
\sigma_{\MMB}(k,n) = \sigma_{\UMMB}(k,n) + r(k,n)+b(k,n)-r'(k,n)-b'(k,n),     
\end{equation} 
because the indicator functions $r'(k,n)$ and $b'(k,n)$ encapsulate the conditions under which we can save a hash in the proof relative to the number of range-node and belt-node ancestors, respectively. 
We define amortized values with formulas similar to~\eqref{eq:amortized}. 
Hence, 
\begin{equation}\label{eq:am-MMB}
\bar{\sigma}_{\MMB}(k)=\bar{\sigma}_{\UMMB}(k)+\bar{r}(k)+\bar{b}(k)-\bar{r'}(k),
\end{equation}
since clearly $\bar{b'}(k)=0$.
Again, we will use the notation and results from Lemma~\ref{lem:Sn}. 
We begin by establishing the periodicity of some of the functions above. 

\begin{lemma}\label{lem:MMB-period}
For any $k\geq 1$, if $d=\lfloor \log_2 (k+1) \rfloor$ then both $r(k+1,n)$ and $b(k+1,n)$ are periodic with period $2^{d+1}$ if $2^d\leq k+1< \frac{3}{2}2^d$, and $2^{d+2}$ if $\frac{3}{2}2^d\leq k+1<2^{d+1}$. 
Furthermore, $r'(k+1,n)$ is periodic with period $2^{d+3}$. 
\end{lemma}
\begin{proof}
Let $d=\lfloor \log_2 (k+1)\rfloor$ and let $M'$ be the mountain that contains the $(k+1)$-th newest leaf. 
Both $r(k+1, n)$ and $b(k+1,n)$, as functions of $n$, depend only on the number of mountains and the positions of range splits to the right of $M'$. 
By Lemmas \ref{lem:Sn} and \ref{lem:MM'}, this information can be decoded from either $(n+1 \mod 2^{d+1})$ or $(n+1\mod 2^{d+2})$, depending on whether $k+1<\frac{3}{2}2^d$, implying the claimed periodicity. 

To prove the second claim, let $k'=2^{d+2}-1$, let $M''$ be the mountain that contains the $(k'+1)$-th leaf (assuming $n$ is large enough), and apply Lemma~\ref{lem:MM'} to them, with $d'=\log_2(k'+1)=d+2$. 
By the Lemma's first claim, $M''$ is always at position $d'-1=d+1$, so it is to the left of $M'$ (which is at position $d$ or $d-1$), and by the third claim, we can decode from $(n+1 \mod 2^{d+3})$ whether there are any range splits between $M''$ and $M'$, and thus whether $M'$ is the leftmost mountain in its range. Hence the claimed periodicity. 
\end{proof}

Next, we make some observations about $r(n,n)$, i.e., the number of mountains in the leftmost range of $\MMB$, as well as conditions for $r'(k,n)$ and $b'(k,n)$ to be $1$.  
The proofs of the next two lemmas are delayed to Appendix~\ref{s:proofs}; see Figure~\ref{fig:leftmost}.

\begin{lemma}\label{lem:rnn}
For any integer $t\geq 0$, we have that: 
\begin{align*}
    \text{for } 4\cdot 2^t &\leq n+1 < 5\cdot 2^t, \quad r(n,n)\geq 2, \quad \text{with sum } \sum_{n=4\cdot 2^t-1}^{5\cdot 2^t-2} r(n,n)=4\cdot 2^t-2, \\
    \text{for } 5\cdot 2^t &\leq n+1 < 6\cdot 2^t, \quad r(n,n)= 2, \\
    \text{for } 6\cdot 2^t &\leq n+1 < 7\cdot 2^t, \quad r(n,n)= 1, \quad \text{and} \\
    \text{for } 7\cdot 2^t &\leq n+1 < 8\cdot 2^t, \quad r(n,n)\geq 2, \quad \text{with sum } \sum_{n=7\cdot 2^t-1}^{8\cdot 2^t-2} r(n,n)=3\cdot 2^t-1.
\end{align*}
\end{lemma}

\begin{lemma}\label{lem:leftmost}
Fix $k\geq 1$, $d=\lfloor \log_2 (k+1) \rfloor$, and let $h$ be the $k$-th newest leaf in $\MMB$. 
For $n=|X|$ in the interval $k\leq n<k+5\cdot 2^d$, $h$ is in the leftmost mountain of its range if and only if $n$ is in one of the three sub-intervals $k\leq n< k+2^d$, $3\cdot 2^{d}\leq n+1< \max\{\frac{7}{2}2^d, k+2\cdot 2^d \}$, or $k+4\cdot 2^d\leq n< k+5\cdot 2^d$. 
Furthermore, $h$ is in the leftmost range within the interval  
$$\begin{cases}
    k\leq n< k+2^{d+1} &\text{for} \quad 2^d\leq k+1 < \frac{3}{2}2^d \\
    k\leq n< k+2^{d+2} &\text{for} \quad \frac{3}{2}2^d\leq k+1 < 2^{d+1}.
\end{cases}$$
\end{lemma}

\begin{figure}[htp]
    \centering
    \includegraphics[width=0.85\textwidth]{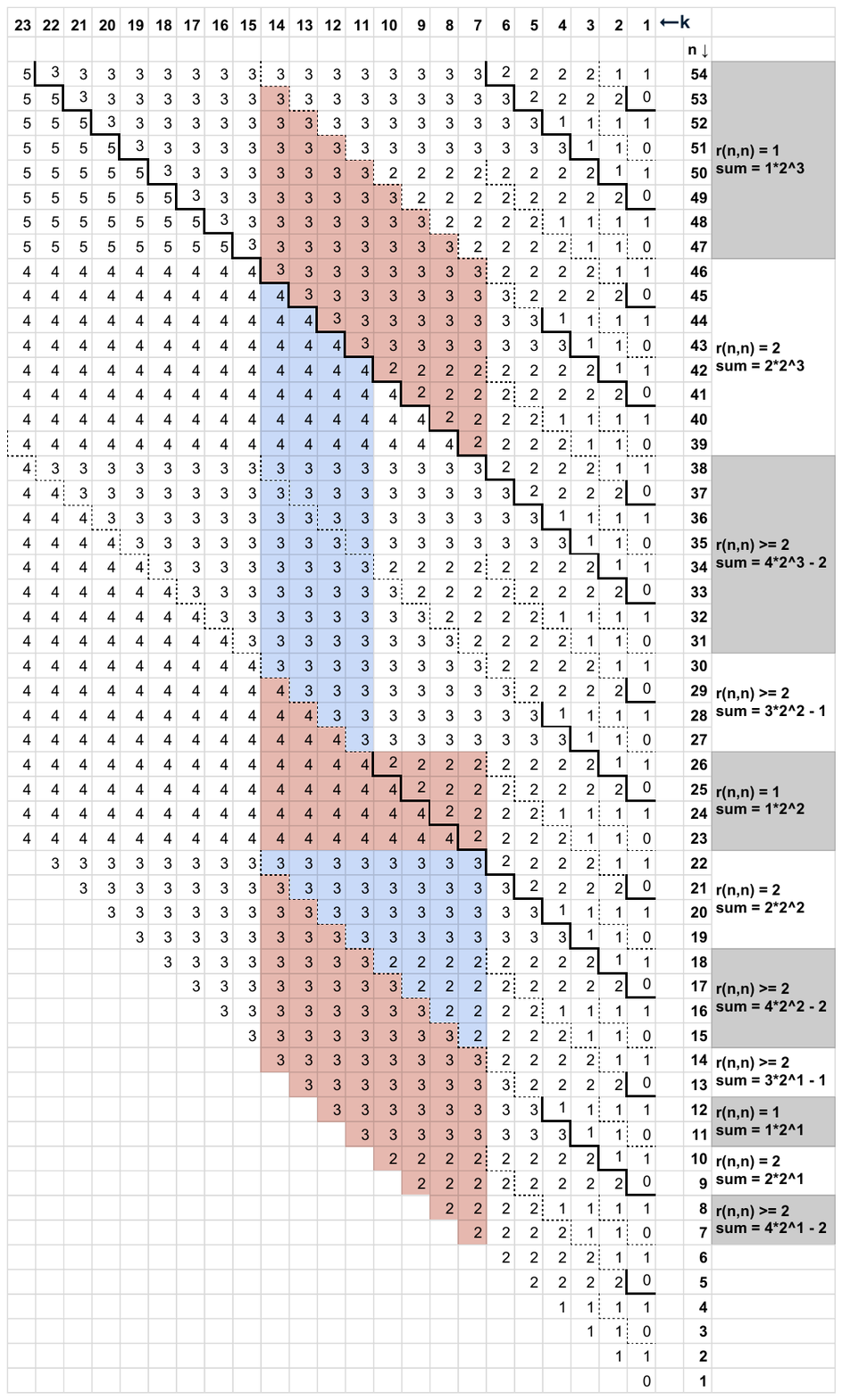}
    \caption{Each cell with coordinates $(k,n)$ is numbered with the height of the mountain that contains the $k$-th newest leaf when $|X|=n$ in $\MMB$. 
    Dotted and bold lines represent mountain divisions and range splits, respectively. 
    The rightmost column represents Lemma~\ref{lem:rnn}. 
    Red and blue cells represent Lemma~\ref{lem:leftmost} for $d=3$, with red cells representing its first statement specifically. } 
    \label{fig:leftmost}
\end{figure}

\begin{corollary}\label{cor:rprime}
For any $k\geq 1$, $\bar{r'}(k)\geq \frac{5}{16}$. 
Furthermore, if $d=\lfloor \log_2(k+1) \rfloor$, the inequality $\frac{1}{2^{d+i}}\sum_{n=k}^{k+2^{d+i}-1} r'(k,n)\geq \frac{5}{16}$ holds for $i=0$, $1$, $2$ and $3$. 
\end{corollary}
\begin{proof}
Fix $k$ and $d$. 
By Lemma~\ref{lem:MMB-period}, $r'(k,n)$ has a period of $2^{d+3}$, so 
$$\bar{r'}(k)=\frac{1}{2^{d+3}}\sum_{n=k}^{k+2^{d+3}-1}r'(k,n).$$ 
By Lemma~\ref{lem:leftmost}, in this sum there are at least three disjoint sub-intervals of $n$ where $r'(k,n)=1$, of sizes at least $2^d$, $\frac{1}{2}2^d$ and $2^d$, respectively. 
Hence, $r'(k)\geq \frac{1}{2^{d+3}}(2^d + \frac{1}{2}2^d+2^d)=\frac{5}{16}$. 

To prove the second claim, we use Lemma~\ref{lem:leftmost} again (see Figure~\ref{fig:leftmost}) to verify that: 

a) In the interval $k\leq n< k+2^{d+2}$, there are two disjoint sub-intervals of $n$ where $r'(k,n)=1$, of sizes at least $2^d$ and $\frac{1}{2}2^d$, respectively, and $\frac{1}{2^{d+2}}(2^d + \frac{1}{2}2^d)=\frac{3}{8}>\frac{5}{16}$. 

b) In the interval $k\leq n< k+2^{d+1}$, there is a sub-interval of $n$ of size $2^d$ where $r'(k,n)=1$, and $\frac{2^d}{2^{d+1}}=\frac{1}{2}>\frac{5}{16}$. 

c) In the interval $k\leq n< k+2^{d}$, we have $r'(k,n)=1$. This completes the proof. 
\end{proof}

\begin{lemma}\label{lem:r}
For any $k\geq 1$, if $d=\lfloor \log_2 (k+1) \rfloor$ then 
$$\bar{r}(k)\leq \begin{cases}
    \frac{9}{4}-\frac{5}{2^{d+1}} & \text{for} \quad 2^d\leq k+1 < \frac{3}{2}2^d \\
    \frac{37}{16} - \frac{3}{2^{d+1}} & \text{for} \quad \frac{3}{2}2^d\leq k+1 < 2^{d+1}.
\end{cases} $$
\end{lemma}
\begin{proof}
Fix an integer $d\geq 1$, and consider first a value of $k$ in the interval $2^d\leq k+1< \frac{3}{2}2^d$. 
By Lemma~\ref{lem:MMB-period}, $r(k,n)$ has a period of $2^{d+1}$, so we can apply equation~\eqref{eq:period} to obtain 
\begin{align*}
    \bar{r}(k) &= \frac{1}{2^{d+1}} \sum_{n=k}^{k+2^{d+1}-1} r(k,n) \\
    &= \frac{1}{2^{d+1}} \left[ \sum_{n=k}^{k+2^{d}-1} r(k,n) + \sum_{n=k+2^d}^{3\cdot 2^d-2} r(k,n) + \sum_{n=3\cdot 2^d-1}^{k+2\cdot 2^{d}-1} r(k,n) \right] \\
    &\leq \frac{1}{2^{d+1}} \left[ \sum_{n=k}^{k+2^{d}-1} r(n,n) + \sum_{n=k+2^d}^{3\cdot 2^d-2} (r(n,n)-1) + \sum_{n=3\cdot 2^d-1}^{k+2\cdot 2^{d}-1} r(n,n) \right] \\
    &= \frac{1}{2^{d+1}} \left[\sum_{n=k}^{k+2^{d+1}-1} r(n,n) -2\cdot 2^d +k+1 \right]=:f(k),
\end{align*}
where the inequality holds because $r(k,n) \leq r(n,n)$ for each value of $n$ in the sum, by the second statement of Lemma~\ref{lem:leftmost}, and this inequality is tight for the first and third sub-intervals, and loose in the second sub-interval, by the first statement. 
Let the bound above be $f(k)$. 
Using Lemma~\ref{lem:rnn}, we can check that this expression, evaluated at $k^*=2^d-1$, gives   
\begin{align*}
f(k^*)=\frac{1}{2^{d+1}}\left[ \sum_{n=2^{d}-1}^{3\cdot 2^{d}-2}r(n,n) -2^d \right]
= \frac{1}{2^{d+1}}\left[\frac{11}{2}2^{d}-5-2^d\right]=\frac{9}{4}-\frac{5}{2^{d+1}}.  
\end{align*}
We claim that this evaluation is maximal; indeed, for a higher value of $k$, the difference is 
\begin{align*}
f(k)-f(k^*) &=\frac{1}{2^{d+1}}\left(k-k^*+\sum_{n=k^*}^{k-1} [r(n+2^{d+1}, n+2^{d+1})-r(n,n)]\right) \\
&\leq \frac{1}{2^{d+1}}\left(k-k^*+\sum_{n=k^*}^{k-1} [-1]\right)=0,
\end{align*}
because, for any $n$ in the interval $2^d\leq n+1<k+1<\frac{3}{2}2^d$, we have $r(n+2^{d+1}, n+2^{d+1})= 1$ and $r(n,n)\geq 2$, by Lemma~\ref{lem:rnn}; see Figure~\ref{fig:leftmost}.  
Putting things together, we have that 
$$\bar{r}(k)\leq f(k) \leq f(k^*) =\frac{9}{4}-\frac{5}{2^{d+1}}.$$

Consider now a value of $k$ in the interval $\frac{3}{2}2^d\leq k+1< 2^{d+1}$. 
By Lemma~\ref{lem:MMB-period}, $r(k,n)$ has a period of $2^{d+2}$, and $r(k,n)\leq r(n,n)$ for the first $2^{d+2}$ values of $n$ by Lemma~\ref{lem:leftmost}, where, again, we can locate the exact sub-intervals where this inequality is tight: 
\begin{align*}
    \bar{r}(k) &= \frac{1}{2^{d+2}} \sum_{n=k}^{k+2^{d+2}-1} r(k,n) \\
    &= \frac{1}{2^{d+2}} \left[ \sum_{n=k}^{k+2^{d}-1} r(k,n) + \sum_{n=k+2^d}^{3\cdot 2^d-2} r(k,n) + \sum_{n=3\cdot 2^d-1}^{k+2\cdot 2^{d}-1} r(k,n) + \sum_{n=k+2\cdot 2^d}^{k+4\cdot 2^d-1} r(k,n)  \right] \\
    &\leq \frac{1}{2^{d+2}} \left[ \sum_{n=k}^{k+2^{d}-1} r(n,n) + \sum_{n=k+2^d}^{3\cdot 2^d-2} (r(n,n)-1) + \sum_{n=3\cdot 2^d-1}^{k+2\cdot 2^{d}-1} r(n,n) + \sum_{n=k+2\cdot 2^d}^{k+4\cdot 2^d-1} (r(n,n) -1) \right] \\
    &= \frac{1}{2^{d+2}} \left[\sum_{n=k}^{k+2^{d+2}-1} r(n,n) -4\cdot 2^d +k+1 \right]=:f(k).
\end{align*}
Let the bound above be $f(k)$. 
By Lemma~\ref{lem:rnn}, its evaluation at $k^*=\frac{7}{4}2^d-1$ is 
\begin{align*}
f(k^*)=\frac{1}{2^{d+2}}\left[ \sum_{n=\frac{7}{4}2^{d}-1}^{\frac{23}{4} 2^{d}-2}r(n,n) -\frac{9}{4}2^d \right]
= \frac{1}{2^{d+2}}\left[\frac{45}{4}2^{d}-6-\frac{9}{4}2^d\right]=\frac{9}{4}-\frac{3}{2^{d+1}}. 
\end{align*}
Now, for lower values of $k$, i.e., $\frac{3}{2}2^d\leq k+1 < \frac{7}{4}2^d$, the difference of evaluations is 
\begin{align*}
f(k)-f(k^*) &=\frac{1}{2^{d+2}}\left(k-k^*+\sum_{n=k}^{k^*-1} [r(n,n)-r(n+2^{d+2}, n+2^{d+2})]\right) \\
&= \frac{1}{2^{d+2}}\left(k-k^*+\sum_{n=k}^{k^*-1} [-1]\right)=-\frac{2(k^*-k)}{2^{d+2}}<0,
\end{align*}
because, for any $n$ in the interval $\frac{3}{2}2^d \leq k+1\leq n+1<\frac{7}{4}2^d$, we have $r(n,n)=1$ and $r(n+2^{d+2}, n+2^{d+2})=2$ by Lemma~\ref{lem:rnn}; see~Figure~\ref{fig:leftmost}.  
And for higher values:  
\begin{align*}
f(k)-f(k^*) &=\frac{1}{2^{d+2}}\left(k-k^*+\sum_{n=k^*}^{k-1} [r(n+2^{d+2}, n+2^{d+2})-r(n,n)]\right) \\
&\leq \frac{1}{2^{d+2}}\left(k-k^*+\sum_{n=k}^{k^*-1} [0]\right)=\frac{k-k^*}{2^{d+2}}<\frac{\frac{1}{4}2^d}{2^{d+2}}=\frac{1}{16},
\end{align*}
because, for any $n$ in the interval $\frac{7}{4}2^d \leq n+1< k+1< 2^d$, we have $r(n+2^{d+2}, n+2^{d+2})=2$ and $r(n,n)\geq 2$, again by Lemma~\ref{lem:rnn}. 
Putting things together: 
$$\bar{r}(k)\leq f(k) \leq f(k^*)+\frac{1}{16} =\frac{9}{4}-\frac{3}{2^{d+1}} +\frac{1}{16} = \frac{37}{16} - \frac{3}{2^{d+1}},$$
as claimed. This completes the proof. 
\end{proof}

\begin{lemma}\label{lem:a-b}
In MMB, for any $k\geq 1$, if $d= \lfloor \log_2 (k+1) \rfloor$ then 
\begin{align*}
\bar{b}(k) &= \begin{cases}
    \frac{3}{8} d + \frac{k+1}{2^{d+1}} +\frac{1}{8} &\text{for} \quad 2^d\leq k+1\leq \frac{3}{2} 2^d \\
    \frac{3}{8} d +\frac{k+1}{2^{d+2}} + \frac{1}{2} &\text{for} \quad \frac{3}{2} 2^d<k+1< 2^{d+1}.
\end{cases}
\end{align*}
\end{lemma}

\begin{proof} 
We prove this formula by induction on $k$, where the base case $\bar{b}(1)=1$ follows because the latest leaf will always be in the rightmost range. 
It can be checked that the induction step corresponds to the following statement: for any $d\geq 1$ and any $k$ with $2^{d}\leq k+1 < 2^{d+1}$, 
\begin{align*}
\bar{b}(k+1) - \bar{b}(k)= \begin{cases}
    \frac{1}{2^{d+1}} & \text{for} \quad 2^d\leq k+1 < \frac{3}{2} 2^d \\
    \frac{1}{2^{d+2}} & \text{for} \quad \frac{3}{2} 2^d \leq k+1 <2^{d+1}.
\end{cases}    
\end{align*} 
If $M$ and $M'$ are the mountains that contain the $k$-th and $(k+1)$-th newest leaves, respectively, the difference $b(k,n) - b(k-1,n)$ is either zero or one, and its average $\bar{b}(k) - \bar{b}(k-1)$ is the probability that $M\neq M'$ and there is a range split between them. 
It follows from Lemma~\ref{lem:MM'} that this probability is either $1/2^{d+1}$ or $1/2^{d+2}$. 
\end{proof}

\begin{lemma}\label{lem:a-mmb}
In $\MMB$, for any $k\geq 1$, if $d=\lfloor \log_2(k+1) \rfloor$ then 
\begin{align*}
\bar{\sigma}_{\MMB}(k) &\leq \begin{cases}
    \frac{11}{8}d+\frac{4(k+1)-5}{2^{d+1}}+\frac{1}{16} &\text{for } \quad 2^d\leq k+1< \frac{3}{2}2^d \\
    \frac{11}{8}d+\frac{3(k+1)-6}{2^{d+2}}+2 &\text{for } \quad \frac{3}{2}2^d\leq k+1< 2^{d+1}
\end{cases} \\
&\leq \frac{11}{8}\log_2 \left(\frac{k+1}{3} \right) +\frac{9}{2}-\frac{9}{4(k+1)}.
\end{align*}
\end{lemma}
\begin{proof}
Fix values of $k$ and $d$. By Equation~\ref{eq:am-MMB}, Lemmas \ref{lem:aUMMB}, \ref{lem:a-b} and \ref{lem:r}, and Corollary~\ref{cor:rprime},  
\begin{align*}
\bar{\sigma}_{\MMB}(k) &\leq \left( d+\frac{3(k+1)}{2^{d+1}} -2 \right)+\left( \frac{3}{8}d +\frac{k+1}{2^{d+1}}+\frac{1}{8} \right) +\left(\frac{9}{4}-\frac{5}{2^{d+1}}\right)-\frac{5}{16} \\
&= \frac{11}{8}d+\frac{4(k+1)-5}{2^{d+1}}+\frac{1}{16} 
\end{align*}
for $2^d\leq k+1<\frac{3}{2}2^d$; and for $\frac{3}{2}2^d\leq k+1< 2^{d+1}$, we have 
\begin{align*}
\bar{\sigma}_{\MMB}(k) &\leq \left( d+\frac{k+1}{2^{d+1}} -\frac{1}{2} \right)+\left( \frac{3}{8}d +\frac{k+1}{2^{d+2}}+\frac{1}{2} \right) +\left(\frac{37}{16}-\frac{3}{2^{d+1}}\right)-\frac{5}{16} \\
&= \frac{11}{8}d+\frac{3(k+1)-6}{2^{d+2}}+2,  
\end{align*}
proving the first claimed bound, which we name $f(k)$.  
Next, by the fact that $f(k)$ is a piecewise linear function, and the second claimed bound, $g(k)$, is a concave function, it suffices to prove that $f(k)\leq g(k)$ at the values of $k$ where $f(k)$ changes its linear behavior, i.e., at $k^*+1\in\{2^d, \frac{3}{2}2^d, 2^{d+1}\}$. 
At $k^*+1=2^d$, we have that $d=\log_2(k^*+1)$, so 
\begin{align*}
f(k^*)&=\frac{11}{8}d +\frac{4(k^*+1)-5}{2^{d+1}}+\frac{1}{16} = \frac{11}{8}\log_2(k^*+1)+\frac{4(k^*+1)-5}{2(k^*+1)}+\frac{1}{16} \\
&= \frac{11}{8}\log_2(k^*+1) +\frac{33}{16}-\frac{5}{2(k^*+1)} < \frac{11}{8}\log_2 \left(\frac{k^*+1}{3} \right) +\frac{9}{2}-\frac{9}{4(k^*+1)},
\end{align*} 
which is $g(k^*)$, as $\frac{33}{16}< \frac{9}{2}-\frac{11}{8}\log_2 3$ and $\frac{5}{2}>\frac{9}{4}$. 
Next, at $k^*+1=\frac{3}{2}2^d$, $d=1+\log_2\left(\frac{k^*+1}{3}\right)$, and it can be checked that $\lim_{k\rightarrow (k^*)^- }f(k)\leq f(k^*)$, so it suffices to evaluate the second line of function $f(k)$ at $k^*$: 
\begin{align*}
f(k^*)&=\frac{11}{8}d +\frac{3(k^*+1)-6}{2^{d+2}}+2 = \frac{11}{8}\left[1+\log_2\left(\frac{k^*+1}{3}\right)\right]+\frac{3(k^*+1)-6}{\frac{8}{3}(k^*+1)}+2 \\
&= \frac{11}{8}\log_2\left(\frac{k^*+1}{3}\right) +\frac{9}{2}-\frac{9}{4(k^*+1)} = g(k^*). 
\end{align*} 
Finally, at $k^*+1=2^{d+1}$, we have that $d=\log_2(k^*+1)-1$, so 
\begin{align*}
\lim_{k\rightarrow (k^*)^-}f(k)&=\frac{11}{8}d +\frac{3(k^*+1)-6}{2^{d+2}}+2 = \frac{11}{8}\left[\log_2(k^*+1)-1\right]+\frac{3(k^*+1)-6}{2(k^*+1)}+2 \\
&= \frac{11}{8}\log_2(k^*+1) +\frac{17}{8}-\frac{3}{k^*+1} < \frac{11}{8}\log_2 \left(\frac{k^*+1}{3} \right) +\frac{9}{2}-\frac{9}{4(k^*+1)},
\end{align*} 
which is $g(k^*)$, as $\frac{17}{8}<\frac{9}{2}-\frac{11}{8}\log_2(3)$ and $3>\frac{9}{4}$. 
This completes the proof. 
\end{proof}

To conclude the section, we remark that $\MMB$ outperforms $\FMMR$ in terms of amortized proof bounds; see Figure~\ref{fig:amortized}, where we plot in blue function $\bar{\sigma}_{\FMMR}(k)$, given by Corollary~\ref{cor:aFMMR}, and in black the upper bound on function $\bar{\sigma}_{\MMB}(k)$, given by the previous lemma. 
It is clear from the figure that $\bar{\sigma}_{\MMB}(k) \leq \bar{\sigma}_{\FMMR}(k)$ for $k\geq 6$, and a closer inspection reveals that this inequality holds for $k\geq 2$.

\section{Delayed proofs}\label{s:proofs}

\begin{proof}[Proof of Lemma~\ref{lem:Sn}]
The last statement easily follows from the first three. 
We prove the first three simultaneously by induction on $n$, where the base case for $n=1$ can be easily verified. 
For the induction step, we consider whether or not $n+1$ is a power of two. 

If $n+1$ is a power of two, and $t:=\log_2(n+1)$, then by induction hypothesis $S_{n-1}= (t-2, t-3, \cdots, 1, 0)+(1,1, \cdots, 1,1) = (t-1, t-2, \cdots, 2, 1)$. 
When we append the $n$-th item we add a height-0 mountain to the right end of the list, and there are no mergeable mountains, so the merge step is skipped and the new list of mountain heights is $S_n=(t-1, t-2, \cdots, 1, 0)$, which is what we needed to show. 

If $n+1$ is not a power of two, then $t:=\lfloor \log_2(n+1)\rfloor=\lfloor \log_2 n \rfloor$. Let $(b_t\cdots b_1 b_0)_2$ and $(b'_t\cdots b'_1 b'_0)_2$ be the binary representations of $n+1$ and $n$ respectively. 
As $n$ is \emph{not} of the form $2^t-1$, not all bits $b'_i$ are equal to one. 
Let $j$ be the lowest index such that $b'_j=0$: this index is precisely $j=\nu_2(n+1)$. 
We have $S_{n-1}=(b'_{t-1} + t-1, \cdots, b'_1+1, b'_0+0)$ by induction hypothesis. 
Hence, the mountain at position $j$ in $S_{n-1}$ is either at the right end of the list and of height 0 (if $j=0$), or (if $j\geq 1$) it is of the same height $j$ as the mountain to its right, and these two mountains are the rightmost mergeable pair. 
In either case, after we add the hash of the $n$-th item as a height-0 mountain, the previously mentioned mountain merges with the mountain to its right, and they are both replaced by a mountain of height $j+1$. 
As a result, the new list of mountain heights is 
\begin{align*}
S_{n}&=(b'_{r-1} + r-1, \cdots, b'_{j+1} + j+1, b'_j+j+1, b'_{j-2}+j-2, \cdots,  b'_0+0, 0)\\
&=(b'_{r-1}+r-1, \cdots, b'_{j+1}+j+1, j+1, j-1, \cdots, 1, 0)\\
&=(b_{r-1}+r-1, \cdots, b_{j+1}+j+1, b_j+j, b_{j-1}+j-1, \cdots, b_1+1, b_0+0),
\end{align*}
which is what we needed to show. 
\end{proof}

\begin{proof}[Proof of Lemma~\ref{lem:kth}]
Recall from Lemma~\ref{lem:Sn} that if $S_n=(s_{t-1}, \cdots, s_1, s_0)$ is the mountain height sequence, then $s_i\in\{i, i+1\}$ for each $0\leq i<t$. 
Hence, for any $d\geq 1$, the $d$ rightmost mountains contain at least the most recent
$$\sum_{i=0}^{d-1} 2^{s_i} \geq \sum_{i=0}^{d-1} 2^i = 2^d - 1$$
leaves; or conversely, each of these leaves is in one of the $d$ rightmost mountains, and its mountain is of height at most $s_{d-1}\leq d$. 
By setting $d:= \lceil \log_2 (k+1) \rceil = \lfloor \log_2 k \rfloor+1$, we ensure that the $k$-th most recent leaf is among these leaves.
\end{proof}

\begin{proof}[Proof of Lemma~\ref{lem:recent}]
Consider a scheme that observes recent-proof compatibility with parameter $c>0$, 
and a membership proof $\pi_{x_i\in X_m}$, $0<i\leq m$, that will need to be updated $t$ times during an upcoming $k$-increment, at values $m_0:=m<m_1<m_2<\cdots<m_t\leq m+k$. 
By hypothesis, we have that $m_{j+1} > m_j + c(m_j-i+1)$, or equivalently,
    $$(m_{j+1} - i+1) > (1+c)(m_j-i+1), \quad \text{for each } 0\leq j<t.$$
We obtain therefore that $(m_t - i+1) > (1+c)^{t}(m_0-i+1)$. Solving for $t$: 
    $$t< \log_{1+c} \frac{m_t-i+1}{m_0-i+1} \leq \log_{1+c} \frac{m+k-i+1}{m-i+1} \leq \log_{1+c} (k+1) = O(\log k),$$ 
where we applied the bounds $m_t\leq m+k$ and $m-i\geq 0$. 
This proves low update frequency.
\end{proof}

\begin{proof}[Proof of Lemma~\ref{lem:hash-F}] 
We need at most one hash computation in the bottom layer, namely the merge peak (if there is a merge step). 
There are $t=\lfloor \log_2(n+1) \rfloor$ peaks, by Lemma~\ref{lem:Sn}, and an equal number of range nodes in F-MMB, so there are at most $t$ hash computations needed in the top layer. 
And again by Lemma~\ref{lem:Sn}, when appending the $n$-th item the merge peak is at position $j=\nu_2(n+1)$ in the peak list, when counting from zero from the right. 
Hence, exactly $\nu_2(n+1)+1$ range nodes need to be created or updated, and this expression has a long-term average value of $2$ by Lemma~\ref{lem:nu}. 
This completes the proof.
\end{proof}

\begin{proof}[Proof of Lemma~\ref{lem:hash-d}] 
We assume the worst (and most common) case where a merge takes place; hence we need one hash computation in the bottom layer, namely the merge peak.

If $n$ is odd, the new leaf is the only zero-height mountain, so it does not participate in the merge.
By Lemma~\ref{lem:close}, the merge peak is at the end of the second-to-last range, while all mountains to its right --~including the new leaf~-- form a single range. 
So we need to update one range root and one belt node above each of the merge peak and the new leaf, giving us $4$ hash computations in the bagging layers, and $5$ in total. 

If $n$ is even, the new leaf is not the only zero-height mountain, so it participates in the merge. 
This merge peak is at the very end of the list, so we only need to update the range and belt nodes above it, giving us $2$ hash computations in the bagging layers, and $3$ in total. 

We conclude that on average we need $4$ hash computations in total.
\end{proof}

\begin{proof} [Proof of Lemma~\ref{lem:nu}] 
Fix a bound $N$. We can convert the sum $\sum_{n=1}^N \nu_2(n)$ into a product: 
$$\sum_{n=1}^N \nu_2(n) = \nu_2\left(\prod_{n=1}^N n\right)=\nu_2(N!), $$ 
because $\nu_2(n)$ is a completely additive function. Using Legendre's formula, we obtain
$$\nu_2(N!)=\sum_{i=1}^\infty \lfloor N/2^i \rfloor = \sum_{i=1}^\infty N/2^i - O(\log N)=N-O(\log N).$$
Thus, the average value is $\frac{1}{N}\sum_{n=1}^N \nu_2(n)=1-O(\frac{\log N}{N})\rightarrow 1$, which proves the claim.
\end{proof}

\begin{proof}[Proof of Lemma~\ref{lem:diffU}] 
From Lemmas \ref{lem:aUMMR} and \ref{lem:aUMMB}, it can be checked that the difference of functions $\bar{\sigma}_{\UMMR}(k)-\bar{\sigma}_{\UMMB}(k)$ is $\frac{1}{2}$ for $k=1$ and $\frac{3}{4}$ for $k=2$, which observe the claimed inequality. 
For any integer $d\geq 2$, it can also be checked that for $k=2^d-1$ and $k=2^d$, this difference is $1+\frac{k-3}{2^{d+1}}$, which observes the claimed inequality loosely. 

For any other case (i.e., $k\geq 5$ and neither $k$ nor $k+1$ is a power of $2$), if $d=\lfloor \log_2 (k+1) \rfloor$ and $d'=\lceil \log_2 k \rceil$ then $d'= d+1$, and  
\begin{align*}
\bar{\sigma}_{\UMMR}(k) - \bar{\sigma}_{\UMMB}(k)&= \begin{cases}
    2- \frac{k+3}{2^{d+1}} & \text{for} \quad 2^d\leq k+1< \frac{3}{2} 2^d \\
    \frac{1}{2}+\frac{k-1}{2^{d+1}} & \text{for} \quad \frac{3}{2} 2^d\leq k+1< 2^{d+1} . 
\end{cases} 
\end{align*} 
In the first case, condition $ k+1< \frac{3}{2} 2^d$ yields inequality $-\frac{1}{2^{d+1}}< -\frac{3}{4(k+1)}$, so 
$$2- \frac{k+3}{2^{d+1}}< 2-\frac{3(k+3)}{4(k+1)}=\frac{5}{4}-\frac{3}{2(k+1)}.$$ 
Similarly, in the second case, condition $k+1\geq \frac{3}{4} 2^{d+1}$ yields inequality $\frac{1}{2^{d+1}}\geq \frac{3}{4(k+1)}$, so 
$$\frac{1}{2}+ \frac{k-1}{2^{d+1}}\geq \frac{1}{2}+\frac{3(k-1)}{4(k+1)}=\frac{5}{4}-\frac{3}{2(k+1)}.$$ 
This completes the proof.
\end{proof}

\begin{proof}[Proof of Lemma~\ref{lem:rnn}]
We prove the claims by induction on $t$, where the base case $t=0$ corresponds to values $r(3,3)=r(4,4)=r(6,6)=2$ and $r(5,5)=1$, which can be verified by hand. 
Now fix an integer $t\geq 1$, consider a value of $n$ with $4\cdot 2^t\leq n+1<8\cdot 2^t$, and let $(b_{t+2}b_{t+1}b_{t}\cdots b_0)$ be the binary representation of $n+1$. 

Case 1: $4\cdot 2^t \leq n+1 < 5\cdot 2^t$. 
The highest bits in $n+1$ are $(b_{t+2}b_{t+1}b_{t})=100$, so the two leftmost mountains $M$ and $M'$ are in the same range and have heights $t+1$ and $t$, respectively. 
Their range contains $r(n,n)=1+r(n',n')$ mountains, where $n'=n-2^{t+1}$, because the $\MMB$ structure for $|X|=n'$ is exactly the same as that for $|X|=n$, except that $M$ is removed, leaving $M'$ as the leftmost mountain. 
By induction hypothesis, 
\begin{align*}
    \sum_{n=4\cdot 2^t-1}^{5\cdot 2^t-2} r(n,n) &= \sum_{n=4\cdot 2^t-1}^{5\cdot 2^t-2} 1+r(n-2^{t+1},n-2^{t+1}) \\
    &= 2^t + \sum_{n=4\cdot 2^{t-1}-1}^{5\cdot 2^{t-1}-2} r(n,n) + \sum_{n=5\cdot 2^{t-1}-1}^{6\cdot 2^{t-1}-2} r(n,n) \\
    &= 2^t + (4\cdot 2^{t-1} -2) + (2\cdot 2^{t-1}) \\
    &= 4\cdot 2^t -2.
\end{align*}

Case 2: $5\cdot 2^t \leq n+1 < 6\cdot 2^t$. 
The highest bits in $n+1$ are $(b_{t+2}b_{t+1}b_{t})=101| $, so the leftmost range consists of two mergeable mountains of height $t+1$. Hence, $r(n,n)=2$. 

Case 3: $6\cdot 2^t \leq n+1 < 7\cdot 2^t$. 
The highest bits in $n+1$ are $(b_{t+2}b_{t+1}b_{t})=11|0 $, so the leftmost range consists of a single mountain of height $t+2$. Hence, $r(n,n)=1$. 

Case 4: $7\cdot 2^t \leq n+1 < 8\cdot 2^t$. 
The highest bits in $n+1$ are $(b_{t+2}b_{t+1}b_{t})=111 $, so the two leftmost mountains $M$ and $M'$ are in the same range and have heights $t+2$ and $t+1$, respectively. 
Their range contains $r(n,n)=1+r(n',n')$ mountains, where $n'=n-2^{t+2}$, because the $\MMB$ structure for $|X|=n'$ is exactly the same as that for $|X|=n$, except that $M$ is removed, leaving $M'$ as the leftmost mountain. By induction hypothesis, 
\begin{align*}
    \sum_{n=7\cdot 2^t-1}^{8\cdot 2^t-2} r(n,n) &= \sum_{n=7\cdot 2^t-1}^{8\cdot 2^t-2} 1+r(n-2^{t+2},n-2^{t+2}) \\
    &= 2^t + \sum_{n=6\cdot 2^{t-1}-1}^{7\cdot 2^{t-1}-2} r(n,n) + \sum_{n=7\cdot 2^{t-1}-1}^{8\cdot 2^{t-1}-2} r(n,n) \\
    &= 2^t + (1\cdot 2^{t-1}) + (3\cdot 2^{t-1} -1)  \\
    &= 3\cdot 2^t -1, 
\end{align*}
as claimed. This completes the proof.
\end{proof}

\begin{proof}[Proof of Lemma~\ref{lem:leftmost}]
See Figure~\ref{fig:leftmost}. 
We rephrase the statements by first fixing $d$, then $n=|X|$, and then letting $k$ range within the corresponding interval. 
We select longer intervals for $k$ than necessary, thus making stronger claims. 
Fix an integer $d\geq 1$: 
\begin{itemize}
    \item Claim 1.A: For $2^d\leq n+1<  3\cdot 2^d$, the $k$-th newest leaf is \textit{in the (globally) leftmost mountain} as long as $\max \{2^d-2, n-2^d\}< k\leq \min\{2\cdot 2^d-2, n\}\leq n$, i.e., the leftmost mountain contains at least $n-\max\{2^d-2, n-2^d\}=\min\{n+2-2^d, 2^d\}$ leaves. 
    \item Claim 1.B: For $3\cdot 2^d\leq n+1< \frac{7}{2} 2^d$, the $k$-th newest leaf is \textit{in the leftmost mountain of a range} as long as $n-\frac{5}{2}2^{d}\leq 2^d-2<k\leq n$, i.e.,  the $\frac{5}{2}2^d$ leftmost leaves are contained in mountains that are leftmost within their ranges.
    \item Claim 1.C: For $5\cdot 2^d\leq n+1< 7\cdot 2^d$, the $k$-th newest leaf is \textit{in the leftmost mountain of a range} as long as $\max\{2^d-2, n-5\cdot 2^d\}<k\leq \min\{2\cdot 2^d-2, n-4\cdot 2^d\}\leq n-4\cdot 2^d$, i.e., apart from the $4\cdot 2^d$ leftmost leaves, the next $\min\{n-5\cdot 2^d+2, 2^d\}$ leaves are contained in mountains that are leftmost within their ranges. 
    \item Claim 2.A: For $2\cdot 2^d\leq n+1< \frac{7}{2}2^d$, the $k$-th newest leaf is \textit{in the leftmost range} as long as $\max \{2^d-2, n-2\cdot 2^{d}\}< k\leq \min\{2\cdot 2^d-2,n\}\leq n$, i.e., the leftmost range contains at least $\min\{n+2-2^d, 2\cdot 2^{d}\}$ leaves. 
    \item Claim 2.B: For $\frac{7}{2}2^d\leq n+1 < 6\cdot 2^d$, the $k$-th newest leaf is \textit{in the leftmost range} as long as $\max\{\frac{3}{2}2^d-2, n-4\cdot 2^{d}\}< k\leq 2\cdot 2^d-2\leq n$, or in other words, the leftmost range contains at least $\min\{ n+2-\frac{3}{2}2^d, 4\cdot 2^{d}\}$ leaves. 
\end{itemize}

To prove these claims, we fix a value of $d\geq 1$, consider several sub-intervals for $n$, and follow the four cases of the proof of Lemma~\ref{lem:rnn}, for specific values of $t$.

For $2^d\leq n+1< \frac{3}{2}2^d$ ($t=d-2$, cases 1 and 2), the number of leaves in the leftmost mountain is $2^{t+1}=2^{d-1}\geq n+2-2^d$ (because $n+1<\frac{3}{2}2^d$), proving claim 1.A. 

For $\frac{3}{2}2^d \leq n+1< 2\cdot 2^d$ ($t=d-2$, cases 3 and 4), the number of leaves in the leftmost mountain is $2^{t+2}=2^d \geq n+2-2^d$ (because $n+1<2\cdot 2^d$), proving claim 1.A. 

For $2\cdot 2^d \leq n+1< \frac{5}{2} 2^d$ ($t=d-1$, case 1), the leftmost mountain has $2^{t+1}=2^d$ leaves, proving claim 1.A. 
And the leftmost range has at least two mountains with $2^{t+1}+2^t=\frac{3}{2}2^d$ leaves, where $\frac{3}{2}2^d\geq n+2-2^d$ (because $n+1<\frac{5}{2}2^d$), proving claim 2.A. 

For $\frac{5}{2}2^d \leq n+1 < 3\cdot 2^d$ ($t=d-1$, case 2), the leftmost mountain has $2^{t+1}=2^d$ leaves, proving claim 1.A. 
And the leftmost range has at least two mountains with $2^{t+1}+2^{t+1}=2\cdot 2^{d}$ leaves, proving claim 2.A. 

For $3\cdot 2^d\leq n+1 < \frac{7}{2}2^d$ ($t=d-1$, case 3), the leftmost range consists of a single mountain with  $2^{t+2}=2\cdot 2^{d}$ leaves, proving claim 2.A. 
And the mountain to its right, which is leftmost within its range, contains at least $2^t=\frac{1}{2}2^t$ more leaves, proving claim 1.B. 

For $\frac{7}{2}2^d\leq n+1 < 4\cdot 2^d$ ($t=d-1$, case 4), the leftmost range has at least two mountains with $2^{t+2}+2^{t+1}=3\cdot2^d$ leaves, where $3\cdot 2^d > n+2-\frac{3}{2}2^d$, proving claim 2.B.

For $4\cdot 2^d \leq n+1 < \frac{9}{2}2^d$ ($t=d$, case 1), the leftmost range has at least two mountains with $2^{t+1}+2^{t}=3\cdot2^d$ leaves, where $3\cdot 2^d \geq n+2-\frac{3}{2}2^d$ (as $n+1< \frac{9}{2}2^d$), proving 2.B.

For $\frac{9}{2}2^d \leq n+1 < 5\cdot 2^d$ ($t=d$, case 1), the  highest bits in $n+1$ are $(b_{t+2}b_{t+1}b_tb_{t-1})=1001$, so the leftmost range has three mountains with $2^{t+1}+2^t+2^t=4\cdot 2^{d}$ leaves, proving 2.B.  

For $5\cdot 2^d\leq n+1<\frac{11}{2} 2^d$ ($t=d$, case 2), the leftmost range consists of two mountains with $2^{t+1}+2^{t+1}=4\cdot 2^{d}$ leaves, proving claim 2.B. 
And the mountain to their right, which is leftmost within its range, contains at least $2^{t-1}=\frac{1}{2}2^d\geq n-5\cdot 2^d+2$ (because $n+1<\frac{11}{2}2^d$), proving claim 1.C.

For $\frac{11}{2} 2^d\leq n+1< 6\cdot 2^d$ ($t=d$, case 2),  the highest bits in $n+1$ are $b_{t+2}b_{t+1}b_tb_{t-1}=1011$, so the leftmost range consists of two mountains with $2^{t+1}+2^{t+1}=4\cdot 2^{d}$ leaves, proving claim 2.B. 
And the mountain to their right, which is leftmost within its range, contains $2^{t}=2^d$ leaves, proving claim 1.C. 

For $6\cdot 2^d\leq n+1<7\cdot 2^d$ ($t=d$, case 3), the leftmost range consists of a single mountain with $2^{t+2}=4\cdot 2^d$ leaves. And the mountain to its right, which is leftmost within its range, contains at least $2^t=2^d$ leaves, proving claim 1.C. 
This completes the proof.
\end{proof}

\end{document}